%% file: arxiv.tex
\documentclass[final,11pt]{amsart}
\usepackage{amsmath,amssymb,amsfonts} %\usepackage{mathtools}
\usepackage{epsfig} \usepackage{latexsym,nicefrac,bbm}
\usepackage{xspace}
\usepackage{color,fancybox,graphicx,subfigure,fullpage}
\usepackage[top=1in, bottom=1in, left=1in, right=1in]{geometry}
\usepackage{tabularx} \usepackage{hyperref} 
\usepackage{pdfsync}

\newcommand{\parag}[1]{ {\bf \noindent #1}}
\newcommand{\defeq}{\stackrel{\textup{def}}{=}}
\newcommand{\nfrac}{\nicefrac}
\newcommand{\eps}{\epsilon}

\newcommand{\opt}{\mathrm{opt}}

\newcommand{\supp}{\mathrm{supp}}
\newcommand{\rank}{\mathrm{rank}}
\newcommand{\maxx}{M_x}
\newcommand{\mdet}{\mathcal{D}}
\newcommand{\Wb}{\overline{W_J}}
\newcommand{\Wt}{\widetilde{W_J}}
\newcommand{\px}{P_{\mathrm{max}}}

\newcommand{\ene}{\mathcal{E}}
\newcommand{\barr}{\mathcal{B}}
\newcommand{\pot}{\mathcal{V}}
\newcommand{\conv}{\mathrm{conv}}
\newcommand{\dist}{\mathrm{dist}}

\renewcommand{\epsilon}{\varepsilon}

\def\showauthornotes{1} 
 
\def\showdraftbox{1}

\input{macros}

\begin{document}
\title{\bf On a Natural Dynamics for Linear Programming}

\author{Damian Straszak}
\thanks{Damian Straszak, \'{E}cole Polytechnique F\'{e}d\'{e}rale de Lausanne (EPFL)} 

\author{Nisheeth K. Vishnoi}
\thanks{Nisheeth K. Vishnoi, \'{E}cole Polytechnique F\'{e}d\'{e}rale de Lausanne (EPFL)}

\maketitle

\begin{abstract}
In this paper we study dynamics inspired by {\em Physarum polycephalum} (a slime mold) for solving linear programs (\cite{NTY00,IJNT11,JZ12}). These dynamics are arrived at by a local and mechanistic interpretation of the inner workings of the slime mold and a global optimization perspective has been  lacking even in the simplest of instances.  Our first result is an interpretation of the dynamics as an optimization process. We show that Physarum dynamics can be seen as a steepest-descent type algorithm on a certain Riemannian manifold. Moreover, we prove that the trajectories of Physarum are in fact paths of optimizers to a parametrized family of convex programs, in which the objective is a linear cost function regularized by an entropy barrier. Subsequently, we rigorously establish several important properties of solution curves of Physarum. We prove global existence of such solutions and show that they have limits, being optimal solutions of the underlying LP. 
Finally, we  show that the discretization of  the Physarum dynamics is efficient  for a  class of linear programs,   which include unimodular constraint matrices. Thus, together, our results shed some light on how nature might be  solving instances of perhaps the most complex problem in {\bf P}: linear programming.

\end{abstract}

\tableofcontents

\thispagestyle{empty}

\newpage

\setcounter{page}{1}

\section{Introduction}

There is growing evidence that several fundamental processes in nature are inherently computational and are best viewed from the standpoint of computation.
Consider the  case of {\em Physarum polycephalum} (slime mold), a single celled organism which has been the source of much excitement among biologists and computer scientists due to its ability to solve complex optimization problems.
This started with  an experiment \cite{NTY00} that showed the slime mold could solve the shortest path problem on a maze.  Roughly speaking, the slime mold  is a collection of tubes which it uses to transport food across its body and it is in this process that it shows the ability to compute. 
Soon after, the time evolution of Physarum was captured  by mathematical biologists using the language of dynamical systems giving rise to a broad class of dynamics for basic computational problems such as shortest paths, transshipment problems, and linear programs (see \cite{TKN07, IJNT11,JZ12}).

The Physarum dynamics are highly nonlinear and  decentralized and, thus,  there is no a priori reason for them to converge anywhere, let alone to solutions of problems, such as those listed above, 
 where a global optima is sought. However,   experiments and simulations suggest that the solution not only converges, it converges quickly! 
Two sets of questions arise naturally from a computational viewpoint: (1) While the dynamical system gives a mechanistic insight into the workings of a Physarum, it fails to explain what is going on globally. Can we explain Physarum dynamics from an optimization  perspective? (2) Can we rigorously explain what experiments, simulations and partial results about the Physarum dynamics suggest, namely the
 (a) existence of a solution, (b) convergence to an optimal solution and (c) time to convergence of the (continuous \emph{and} discretized) dynamics?

In this paper we study these questions for the Physarum dynamics for linear programming proposed by \cite{JZ12} which generalizes earlier models for shortest path and the transshipment problem.
For question (1) there  are no known answers even in the simplest setting of shortest path:  on the one  hand the Physarum dynamics resemble the gradient descent method; however,  one can show that it is not. On the other hand, 
the dynamics are reminiscent of interior point methods; however, what  makes the Physarum dynamics remarkable is that one can start from outside the feasible region of the linear program and yet converge to the optimal point.
For question (2), the situation is slightly better. For various special cases such as 
shortest path \cite{BMV12, BBDKM13} and the transshipment problem \cite{IJNT11, SV15} 
(a)-(c) have been answered. However, these results heavily rely on  the special structure of the problem. For LPs, the prior results prove convergence {\em assuming} that (i) the solution exists (ii) the feasible region of the LP is bounded and (iii) the optimizer is unique. As we argue later, justifying or removing these ends up being quite hard. Current results do not bound the convergence time.

The main contribution of this paper is to provide answers to all of these questions: We show that Physarum dynamics can be seen as a steepest-descent type algorithm, however, not in Euclidean space, rather in a space endowed with a  {\em Riemannian metric} obtained from an {\em entropy-like function}. 
As a consequence, we show that Physarum dynamics for linear programming are obtained by balancing two forces in this Riemannian manifold: the need to reduce cost and the need to be feasible.
Moreover, we show that the continuous trajectories of Physarum are in fact paths of optimizers to a parametrized family of convex programs, in which the objective is a linear cost function regularized by an entropy barrier. Subsequently, we 
 establish the global existence of  solutions of Physarum dynamics and show that they have limits, being optimal solutions of the underlying LP. 
Finally,  we present a time-bound on a discretization of the dynamics to yield an algorithm that is provably efficient for a  class of linear programs,   which include unimodular constraint matrices.
Our proofs synthesize concepts and tools from several distinct areas such as Riemannian manifolds, convex optimization and dynamical systems. These techniques should be useful in providing  insights into the plethora of   steepest-descent type algorithms  in optimization, machine learning and,  recently, in theory.

 \paragraph{\bf The Physarum Dynamics for Linear Programming.} 

       Consider a linear program in the standard form 
\begin{equation}\label{lp}
\min \; c^\top x \; \; \; \; \mbox{s.t.} \; \; Ax=b, \; \; \;  x\geq 0
\end{equation}
 where $A\in \Z^{m\times n}$, $c\in \Z^{n}_{>0}$ and $b\in \Z^m$ and which has a feasible solution. 
To make the exposition clear we assume that $A$ is full-rank, in other words: $\rank(A)=m$.\footnote{This assumption is not necessary for our results.} 
We now describe the Physarum dynamics for linear programming. 

Consider any vector $x\in \R^n$ with $x>0$ 
and let  $W$ be the diagonal matrix with entries $\nfrac {x_i}{c_i}.$ Let  $L\defeq AWA^\top$ and  $p \in \R^n$ is the solution to $L p = b.$\footnote{One can show that $p$ exists and is unique because of the full-rank assumption on $A$, otherwise one can use the pseudo-inverse.}
Let $q \defeq WA^\top p.$ 
The Physarum dynamics for the linear program given by $(A,b,c)$ then is
\begin{align}\label{eq:dynamics}
\dot{x}= q-x
\end{align}
defined over the domain $\Omega = \R^{n}_{>0}.$ The above is just a simplified notation for the coupled differential equations:
$$ \forall i,  \; \; \frac{dx_i(t)}{dt} =q_i(t) - x_i(t).$$
 This can be also rewritten as
\begin{align}\label{eq:dynamics_matrix}
\dot{x}= W(A^\top L^{-1} b - c).
\end{align}
The dynamical system has an initial condition of the form $x(0)=s$ for some $s>0.$ {\em Note that it is not required  that $s$ is feasible}, i.e., $As=b.$  Often, the exposition and proofs are simpler when $s$ is also feasible. We often refer to the general case as Physarum dynamics with {\em infeasible start} and this special case as with {\em feasible start.}  
Notice that when $A$ is the $|V|\times |E|$ incidence matrix of a graph $G=(V,E),$ then if we set up an electrical network where the edge $i$ has resistance $\nfrac{x_{i}}{c_{i}},$ then $L$ is the corresponding graph Laplacian and $q$ the electrical flow in each edge corresponding to the demand vector $b.$ Viewing the slime mold as a collection of tubes organized in a network where the length of edge $i$ is  $c_{i}$  and cross-section $x_{i}$  and noting that the equations of Poiseuille flows are identical to electrical flows,  \eqref{eq:dynamics} captures the local time evolution of the  {\em cross-section} of each tube.

\vspace{2mm}
\paragraph{\bf Our Contributions.}
The first problem we have at hand is whether the system~\eqref{eq:dynamics} has a solution such that $x(t)>0$ for all $t>0.$ Unfortunately there are no general theorems in dynamical systems to establish this and proving existence can be difficult.
Indeed, this existence problem turns out to be non-trivial for Physarum dynamics and is our first result.

\begin{theorem}[{\bf Existence of Solution}]\label{thm:existence0}
For any initial condition $s\in \R_{>0}^n,$  the Physarum dynamics $\dot{x}=q-x$ has a unique solution $x:[0,\infty)\to \R_{>0}^n$ with $x(0)=s$.
\end{theorem}
\noindent
First we note that existence of a feasible solution is important:  One can prove that whenever the problem~\eqref{lp} is infeasible, the solutions to~\eqref{eq:dynamics} exist only for finite time intervals. 
Further, the problem of existence is  difficult not because the dynamics is in continuous time; the problem remains non-trivial even if we consider a discretization. 
Now that we know that the solution exists, our next set of results establish that no matter where one starts the Physarum dynamics from, as long as it is in the positive orthant, one converges to an optimal solution of the linear program. 

\begin{theorem}[{\bf Convergence to an Optimal Solution}]\label{thm:conv_opt}
For any initial condition $s\in \R_{>0}^n$,  consider the solution $x:[0,\infty)\to \R_{>0}^n$ to the Physarum dynamics with $x(0)=s$. Denote by $x^\star$ an optimal solution of the considered linear program. Then:
\begin{enumerate}
\item $|c^\top x(t) - c^\top x^\star| = O(e^{-\alpha t})$ for some positive $\alpha$, which only depends upon $A,b,c$ and $s$,
\item the limit $x^\infty = \lim_{t\to \infty} x(t)$ exists, is a feasible point and $c^\top x^\infty = c^\top x^\star$.
\end{enumerate}
\end{theorem}
\noindent
A few remarks are in order: (1) The theorem also gives an upper bound on the time to convergence and (2) it proves that limits of trajectories of the Physarum exist even when the optimal solution is {\em not unique}. Given the amount of technical effort required in proving this, one may wonder again if this has something to do with continuous time. The answer is (again) no: Starting with Karmarkar himself \cite{Karmarkar90}, it has been observed that the good convergence properties of   his  algorithm for linear programming \cite{Karmarkar84} arise from the good geometric properties of the set of continuous trajectories which underlie his method \cite{Bayer1989}.

Our next result concerns the computational abilities of a discretization of the Physarum dynamics.  
\begin{theorem}[{\bf Convergence Time of Discrete Physarum Dynamics; see Theorem~\ref{theorem:discrete_conv}}]\label{thm:discreteLP}
Consider the discretization of the Physarum dynamics, i.e., $x(k+1) = (1-h)x(k)+hq(k)$. 
Suppose we initialize the Physarum algorithm with $x(0)=s$, i.e.,  $As=b$ and $M^{-1} \leq s_i\leq M$ for every $i=1,\ldots,n$ and some $M \geq 1$. Assume additionally that $c^\top s\leq M\cdot \opt$.
Choose any $\eps>0$ and\footnote{$C_s = \sum_{i=1}^n c_i$ and $\mdet = \max\{|\det(A')|:A' \mbox{ is a square submatrix of A}\}$.} let $h=\frac{1}{6}\cdot \eps\cdot C_s^{-2} \cdot \mdet ^{-2}$.
Then after $k\defeq O\inparen{\frac{\ln M}{\eps^2 h^2} }$ steps $x(k)$ is a feasible solution with: $\opt \leq c^\top x(k) \leq (1+\eps)\cdot \opt$.
\end{theorem}
\noindent
Thus, when the maximum cost is polynomial and  the maximum sub-determinant of $A$ is polynomially bounded, the discretization of Physarum dynamics is efficient.

Finally, we come to the conceptual question which has resisted an answer even for the shortest path Physarum dynamics: While each equation in \eqref{eq:dynamics} gives us a local update rule, can the Physarum dynamics be obtained from an optimization viewpoint?
We answer this question affirmatively. For preliminaries on Riemannian manifolds see Section \ref{section:preliminaries}.

\begin{theorem}[{\bf Optimization Viewpoint of Physarum Dynamics; see Theorem~\ref{thm:manifold}}]\label{thm:opt}
In case of feasible start one can interpret Physarum as a gradient flow on a Riemannian manifold. In other words it is of the form $\dot{x} = \nabla f(x)$, where $f(x)=c^\top x$ is the objective, but the gradient is computed with respect to a certain Riemannian metric on the feasible region.
In the infeasible case, Physarum dynamics combines two forces; minimizing the objective and aiming for feasibility. Both can be interpreted similarly as descent directions on a manifold.
\end{theorem}
\noindent
We remark that the above mentioned Riemannian metric is induced by the Hessian of the generalized entropy function $\sum_i c_i x_i \ln x_i c_i$ and provides a physical meaning to the dynamics. 
Moreover, using convex programming duality, we can show that the trajectories of Physarum with a feasible start are in fact paths of optimizers to a parametrized family of convex programs, which objective is a linear cost function regularized by an entropy barrier function; see Theorem~\ref{thm:entropy} for a formal statement. Finally, note that the  entropy barrier, unlike the log-barrier function  is not a self-concordant barrier function; yet the trajectories converge as $e^{-\alpha t}$ as Theorem \ref{thm:conv_opt} indicates.

\vspace{2mm}
\parag{Related Work.}
 There are rigorous results for existence, convergence and time for the shortest path problem \cite{BBDKM13}.  Our previous paper \cite{SV15} deals with the discrete Physarum dynamics transshipment problems on graphs and uses the underlying graph structure heavily to obtain versions of Theorem \ref{thm:discreteLP}.  The LP Physarum dynamics was introduced by \cite{JZ12}.   
There is a large body of literature in the mathematical programming community which studies the geometric properties of continuous trajectories of interior point methods such as projective and affine scaling; see \cite{Karmarkar90, Bayer1989, NesterovTodd02}. As mentioned before the Physarum dynamics is none of these methods (it works from an infeasible start as well), but it bears some similarity to primal affine scaling methods, see, e.g., \cite{lagarias1990ii}. However, for affine scaling methods, the convergence rate can be provably bad: $\nfrac 1 t$, see, e.g., \cite{Megiddo89}.

\section{Technical Overview}
We start by presenting an overview of  our result which presents Physarum dynamics from an optimization viewpoint (Theorem \ref{thm:opt}). We  assume \ref{thm:existence0} and \ref{thm:conv_opt} for the moment.
 Denote 
 $$P(x)\defeq q-x=W(A^\top (AWA^\top)^{-1}b-c),$$ so that the Physarum dynamics  becomes $\dot{x}=P(x)$. Knowing that~\eqref{eq:dynamics} is in fact solving an optimization problem one would expect that it implements some gradient-descent-type of procedure. This would mean that $P(x)=\nabla \Phi(x)$, for some function $\Phi(x)$. However, by a simple calculation, it turns out that $P(x)$ is not a gradient of any function. Similarly, $P(x)$ does not represent Newton's method nor steepest descent in any standard norm. To understand Physarum, we need to move from the Euclidean world to Riemannian manifolds. 

We consider the manifold $\Omega=\R_{>0}^n$, to every point $x\in \Omega$ we assign a positive definite matrix $$H(x)\defeq \diag{\nfrac{c_1}{x_1},\ldots, \nfrac{c_n}{x_n}},$$ which defines an inner product $\inangle{u,v}_x = u^\top H(x) v$ on the tangent space at $x$ (which is just $\R^n$ in this case), $\Omega$ becomes then a Riemannian manifold. We are mainly interested in the submanifold $\Omega^f = \{x\in \Omega: Ax=b\}$, which we assume to be nonempty for our discussion. It inherits the Riemannian structure from $\Omega$. Note that the Nash Embedding Theorem guarantees that such a manifold  always has an isometric embedding in an Euclidean space (possibly in much higher dimension). In our case, the embedding can be given by an explicit formula:$$F(x) =2(\sqrt{c_1x_1}, \ldots, \sqrt{c_nx_n})^\top.$$ This gives as a geometric realization of the curvature imposed on $\Omega^f$. 

Let us now see that $P(x)$, when restricted to $\Omega^f$ is just the gradient of the objective $f(x)=c^\top x$ when viewed in the local geometry. To show this, we need to argue that it satisfies $$f(x+h)-f(x)-\inangle{h,P(x)}_x = o(\norm{h}_x)$$ over $h$ from the tangent space at $x$, i.e., from $T_x(\Omega^f) = \{h: Ah=0\}$. We calculate:  
\begin{align*}
f(x+h)-f(x)-\inangle{h,P(x)}_x &= c^\top h - h^\top H(x) P(x) \\
&= c^\top h - h^\top (A^\top L^{-1} b-c)\\
& = (Ah)^\top L^{-1} b=0.
\end{align*}
\noindent
If $x$ lies outside of the feasible region a similar interpretation of $P(x)$ is possible. However in this case $P(x)$ decomposes into two parts: $P_f(x)$ which is the feasibility direction and $P_o(x)$ which aims for optimality. Both can be motivated as descent directions with respect to the above defined Riemannian structure on $\Omega$.

The next result gives another interpretation of the Physarum dynamics in the feasible region. Namely let us define $x(\mu)$, for $\mu> 0$ to be the minimizer of $c^\top x + \mu^{-1}f(x)$ over $x\in \inbraces{x:Ax=b,x\geq 0}$, where $f$ is the following entropy-like function 
$$f(x) \defeq \sum_{i=1}^n c_i x_i \ln (c_i x_i).$$ It turns out that using elementary convex programming techniques such as duality and KKT conditions, we can prove that $x(\mu)$ satisfies the Physarum dynamics. This is also related to the fact that $\nabla^2 f(x) = H(x)$, which at the same time demonstrates that $(\Omega, \inangle{\cdot, \cdot}_x)$ is a Hessian manifold; see Section~\ref{sec:entropy}.

We now give an overview of the proofs of  the set of results that establish existence, convergence and complexity bounds on the Physarum dynamics (Theorems \ref{thm:existence0}, \ref{thm:conv_opt}, \ref{thm:discreteLP}).  We begin with Theorem~\ref{thm:existence0}. It asserts that global solutions $x:[0,\infty) \to \Omega$ of the Physarum dynamics indeed exist. This theorem is nontrivial because by standard existence-uniqueness theorems we can only obtain solutions defined on some tiny intervals around $0$. To extend those solutions further we need to show that they cannot leave the domain $\Omega$ in finite time. More precisely, if $x:[0,T)\to \Omega$ is a solution, we need to show that the limit point $x(T) = \lim_{t\to T} x(t)$ exists and belongs to $\Omega$. The main concern is that potentially $x(T)$ may end up on the boundary of $\Omega$, i.e., have a zero at some coordinate. Of course, we expect some entries of $x(t)$ to vanish with $t\to \infty$, what we need to show is that this happens in a controlled manner.

Let us take a look at the dynamics: $$\dot{x_i} = q_i - x_i = \frac{x_i}{c_i}(a_i^\top p - c_i),$$ where $a_i$ is the $i$-th column of $A.$  If we could show that $a_i^\top p$ is uniformly bounded then this would imply that $\dot{x_i} \geq -\alpha x_i$ for some constant $\alpha$, and, by the Gronwall lemma (see, e.g.,~\cite{perko2001differential}) $x_i(t) \geq e^{-\alpha t} x_i(0)>0$. Indeed  Lemma~\ref{lemma:key}, which shows that $|a_i^\top p|$ is uniformly bounded over all $x>0$ such that $Ax=b$, rescues us. At first glance Lemma~\ref{lemma:key} may seem a bit technical, but in a sense it captures exactly the property of the Riemannian space, which is responsible for the well behavior of trajectories. 
This gives a proof of existence in the special case when the initial point $x(0)$ satisfies $Ax(0)=b$. (Assuming $Ax(0)=b$ one can show that $Ax(t)=b$ for every $t$, since $\dot{x}$ is tangent to this affine subspace.)\footnote{For the case of feasible start, one can also give another proof of existence via Theorem~\ref{thm:entropy}.} Unfortunately, the case of infeasible start turns out to be harder. An analogue of Lemma~\ref{lemma:key} does not hold if we let $x$ vary over $\Omega$.

We proceed by analyzing a certain barrier function $\sum_{i=1}^n c_i y_i \ln x_i$, with $y$ being any feasible solution to~\eqref{lp}. By analyzing its derivative, we are able to prove that on every finite interval $[0,T)$ this function is uniformly lower-bounded, which essentially means that on every finite interval there exists some $\delta>0$ such that $\delta \cdot y \leq x_i(t)$. This implies further that $x_i(t)$ is lower-bounded by a positive constant, on every finite interval, but only for $i\in \supp(y)$. Of course we may repeat this argument multiple times, with different $y$'s. This however is not enough:  it is a common case for linear programs that there is an index $i\in \{1,2,\ldots, n\}$, such that for every feasible solution $y$, $y_i=0$. Therefore we need to use a different argument for the remaining indices. 

To complete the argument we combine Lemma~\ref{lemma:key} with the above mentioned fact that on every finite interval, $x_i(t)\geq \delta y$ for some $\delta>0$ and some feasible solution $y$. Applying Lemma~\ref{lemma:key} to $y$, we can essentially extract a uniform upper bound on $|a_i^\top p(t)|$ over a fixed finite interval. This, again by the Gronwall lemma, allows us to conclude positivity. In the process of the proof we need to argue that $\norm{q(t)}$ and $\norm{x(t)}$ remain bounded uniformly over all $t$. For this, Lemma~\ref{lemma:key} is again a starting point. 

After establishing existence we study the limiting behavior of the solutions to~\eqref{eq:dynamics}. Our results are summarized in Theorem~\ref{thm:conv_opt}. The first result claims that $$\abs{c^\top x(t) - c^\top x^\star}\leq R\cdot e^{-\nu t},$$ where $R,\nu>0$ are constants depending only on $A,b,c,x(0)$ (we provide explicit bounds in Theorem~\ref{thm:exp}). Furthermore, one can show that $z(t) = x(t) - e^{-t} x(0)$ satisfies $Az(t) = b$. In other words, $x(t)$ approaches feasibility with $t\to \infty$. 

To prove exponential convergence to the optimal value we start by showing that for every $t$ there is some $y(t)$, feasible for~\eqref{lp}, which is exponentially close to $x(t)$. Towards this, we cannot directly use $z(t)$, because it may happen that $z_i(t)<0$ for some $i$. Instead, we use a well known fact that if a point $x$ violates constraints defining a polytope up to an additive error $\eps$ then its distance to this polytope is at most $N\cdot \eps$ where $N$ depends only on the description size of the polytope (though exponentially). 

To proceed, we need to exploit the structure of the feasible region $P=\{x:Ax=b,x\geq 0\}$. We write $P$ as a Minkowski sum $H+C$, where $H$ is a polytope and $C$ is a polyhedral cone. We extract vertices $V$ from $H$ and spanning vectors $R$ from $C$. We study the decomposition of $y(t)$ into 
$$y(t) = \sum_{v\in V} \lambda_v(t) v + \sum_{r\in R} \mu_r(t) r,$$ with $\lambda(t), \mu(t) \geq 0$ and $\sum_{v\in V} \lambda_v(t)=1$. By studying potential functions of the kind $\sum_{i=1}^n c_i v_i \ln x_i(t)$ for a fixed vertex $v\in V$ it can be shown that if $v$ is non-optimal then $\lambda_v(t) \to 0$ exponentially fast. Similarly, for every $r\in R$ we have $\mu_r(t) \to 0$. This is then used to deduce the result.

Let us now discuss the second result related to the asymptotic behavior of solutions to~\eqref{eq:dynamics}. It says that if $x:[0,\infty)\to \Omega$ is such a solution then $x(t)$ has a limit $x^\infty$ with $t\to \infty$. Furthermore $x^\infty$ is an optimal solution to~\eqref{lp}. The second part of this result follows easily once the first is established. If the set of optimal solutions to~\eqref{lp} has only one element $x^\star$, then it is easier to show that $x(t) \to x^\star$, this basically follows from the fact that $c^\top x(t)$ tends to the optimal value and $x(t)$ approaches the feasible region. However if the optimal solution is not unique, then convergence of $x(t)$ is far from clear; $x(t)$ could  oscillate around the optimal set, without converging to a single point. A proof that this does not happen gives us evidence that Physarum dynamics behaves nicely.

To prove convergence to a limit it is enough to show that $$\int_{0}^\infty \norm{\dot{x}(t)}dt <\infty .$$  For this, it suffices to show $\norm{\dot{x}(t)} = O(e^{-\eps t})$ for some $\eps>0$. We will now focus on this task.
Recall that $\dot{x}_i = \frac{x_i}{c_i}\inparen{a_i^\top p -c_i}$. Using similar tools as in the proof of existence one can show that $|a_i^\top p(t)|$ is uniformly bounded over all $t\in [0,\infty)$. Hence, if for some $i\in \{1,2,\ldots, n\}$ we have $x_i(t) = O(e^{-\eps t})$ then it follows easily that $|\dot{x}_i(t)|=O(e^{-\eps t})$.  Denote the set of all such $i$'s by $N$ and its complement by $J$. We can then prove that $J$ is the union of supports of all optimal solutions to~\eqref{lp}. This in turn implies that the subvector $x_J(t)$ behaves in a sense more stably. We show in particular, that for some optimal solution $f$, such that  $\supp(f)=J$ and some $\eps>0$, we have $\eps f \leq x(t)$ uniformly over $t\in [0,\infty)$. Furthermore the key Lemma~\ref{lemma:continuity} establishes some form of continuity of $a_i^\top p(t)$ for $i\in J$. It implies that $|a_i^\top p(t) - c_i| = O(e^{-\eps t})$, which concludes the proof. 

The last theorem we discuss is Theorem~\ref{thm:discreteLP}. It shows that a discretization of the Physarum dynamics is possible, provided that the step length is small enough. The crucial parameter which controls the step length is  $\norm{A^\top p}_\infty$; if one is able to upper bound it over all possible $x$ by some number $\px$  then the step length of roughly $\px^{-2}$ guarantees convergence. In the discretization we choose the starting point $x(0)$ to be feasible, by which $x(k)$ will remain feasible for every $k$. Therefore, by Corollary~\ref{cor:pot_diff} one can take $\px=\mdet \cdot C_s$. The analysis of the discretization follows along the same lines as was presented in \cite{SV15}. Let us discuss it briefly. One can show that $c^\top x(k)$ is decreasing with $k$. However, we cannot guarantee a large drop of this value at every step. Instead we introduce another potential function $$\barr(k) \defeq \sum_{i=1}^n c_i x_i^\star \ln x_i(k)$$ (where $x^\star$ is any optimal solution to~\eqref{lp}) and show that at every step when $c^\top x(k) - c^\top x(k+1)$ is small, $\barr(k)$ grows by a considerable amount. Moreover one can show that $\barr(k)$ is uniformly upper-bounded. By combining $c^\top x(k)$ and $\barr(k)$ into a single potential function $\phi(k)$ we can easily track the progress of our algorithm

\parag{Organization of The Paper.} The remaining part of the paper is structured as follows. In Section~\ref{section:preliminaries} we explain our notation and give a brief introduction to Riemannian manifolds. Section~\ref{section:manifolds} discusses the optimization viewpoint on the Physarum dynamics. In Section~\ref{section:existence} existence of solutions to~\eqref{eq:dynamics} is established. Section~\ref{section:convergence} is devoted to the study of limiting behavior of Physarum integral curves, convergence results are established. The last Section~\ref{section:discretization} discusses the discretization of the dynamics.  Appendix~\ref{section:dynamical} is an introduction to dynamical systems, at the end a general theorem is proved, which is used in the main body to establish existence of solutions.

\section{Preliminaries}\label{section:preliminaries}
\parag{Notation.}
In the paper we often use the convention that for a vector $x\in \R^n$, the capitalized version $X$ denotes the diagonal matrix $\diag{x_1, \ldots, x_n}$. Whenever a scalar function $f:\R\to \R$ is applied to a vector $x\in \R^n$ the component-wise application of $f$ to $x$ is meant, i.e. $f(x)\defeq (f(x_1),\ldots, f(x_n))^\top$. We use the notation $[n]$ to denote the set $\{1,2,\ldots,n\}$. The columns of the matrix $A$ are denoted by $a_i$, $i\in [n]$.
Whenever the objects $W=\diag{\frac{x_1}{c_1}, \ldots, \frac{x_n}{c_n}}$, $L=AWA^\top$, $p=L^{-1}b$ or $q=WA^\top p$ appear, they should be understood as computed with respect to a point $x$, which will be clear from the context. We will often refer to $q(t)$ as to $q$ computed w.r.t. $x(t)$.

\noindent We use the standard norms for $u\in \R^n$: 
\begin{itemize}
\item $\norm{u}_1\defeq\sum_{i=1}^n |u_i|,$ 
\item $\norm{u}_2 \defeq\inparen{\sum_{i=1}^n u_i^2}^{\frac{1}{2}},$
\item $ \norm{u}_\infty \defeq \max\inparen{|u_1|, |u_2|,\ldots, |u_n|}.$ 
\end{itemize}
By $\supp(u)$ we denote the set $\{i\in [n]: u_i\neq 0\}$.
The linear program \eqref{lp} is fixed for the whole paper, we assume the feasible region $\{x:Ax=b,x\geq0\}$ to be nonempty. 

\parag{Parameters of Physarum.}
The following two parameters are often used in statements of quantitative bound regarding convergence time:
\begin{itemize}
\item $C_s = \sum_{i=1}^n c_i$,
\item $\mdet=\max\{|\det(A')|:A'\mbox{ square submatrix of }A\}$.

\end{itemize}

\parag{Riemannian Manifolds.}
All manifolds we deal with can be seen as open subsets of Euclidean spaces, possibly embedded in Euclidean spaces of higher dimension. Some examples of those are $\R^n$, $ \R^n_{>0}$, open polyhedra $\{x\in \R^n: Ax=b,x>0\}$ or a sphere $\{x \in \R^n: \norm{x}_2=1\}$. If $M$ is such a manifold then at every point $x\in M$ one can define a tangent space $T_xM$ in a natural way. A tangent space is equipped with a natural linear structure. For example, the tangent space to any point $x\in M=\{x:Ax=b\}$ is $T_xM=\{x: Ax=0\}$. Given a manifold $M$ we can endow it with a Riemannian structure by assigning to every point $x$ a scalar product $\inangle{\cdot, \cdot}_x$ on $T_xM$ in such a way that $\inangle{\cdot, \cdot}_x$ ``varies smoothly with $x$''.\footnote{One can make this statement precise, but our intention is rather to give some intuitions instead of precise definitions.} When considering a function $F:M\to \R$ on a Riemannian manifold, the gradient of $F$ is defined with respect to the local inner product. Hence in particular, one can define gradient-descent algorithms based on a chosen Riemannian structure of the space, those are typically very different to algorithms based on the standard Euclidean geometry.

\section{Physarum as Steepest Descent on a Manifold}\label{section:manifolds}
In this section we interpret the Physarum dynamics as steepest descent on a certain manifold. We start by equipping the positive orthant with a Riemannian structure. In particular we formalize Theorem \ref{thm:opt} in Theorem \ref{thm:manifold} and prove it. The entropy barrier interpretation is given in Subsection~\ref{sec:entropy}

\subsection{The Riemannian Structure and its Origin}
We focus on two domains: $$\Omega = \R_{>0}^n \quad  \mbox{and}\quad  \Omega^f = \inbraces{x\in \Omega : Ax=b}.$$ Note that $\Omega$ and $\Omega^f$ are manifolds of dimension $n$ and $m$ respectively. We consider a Riemannian structure $\inangle{\cdot, \cdot}_x$ on $\Omega$ (which automatically induces a Riemannian structure on the submanifold $\Omega^f \subseteq \Omega$) given as a family of positive definite matrices $H(x)=CX^{-1}$, for $x\in \Omega$. This gives rise to the following inner product on the tangent space at $x$:
$$\inangle{u,v}_x \defeq u^\top H(x) v.$$
We present two interpretations of where does this structure come from. The first one is related to the following entropy-like function $h:\Omega \to \R$
$$h(x) = \sum_{i=1}^n c_i x_i \ln(c_i x_i) - \sum_{i=1}^n c_ix_i.$$
Note that:
\begin{equation}
\nabla h(x) = C \ln(Cx) \qquad \quad \mbox {and} \quad \qquad \nabla^2 h(x) = CX^{-1}
\end{equation}
hence $H(x)= \nabla^2 h(x)$ is an example of a so called {\it Hessian Metric}.   

\noindent
The second interpretation is based on the analysis of the following map: $F:\Omega \to \Omega$
$$F(x) = 2\sqrt{Cx}.$$
We will now explain that $F$ can be seen as an isometry between our Riemannian manifold $(\Omega, \inangle{\cdot, \cdot}_x)$ and $\Omega$ with the standard inner product $\inangle{\cdot, \cdot}$ at every point. Recall that the Jacobian of $F$\footnote{Jacobian is the multidimensional analogue of a derivative. In the standard coordinate system it can be expressed as a matrix of partial derivatives.}, $J(x)=(CX^{-1})^{1/2}$ acts as a linear map $J(x): T_x(\Omega) \to T_x(\Omega)$, i.e.:
$$J(x)(u) = J(x)u = \inparen{CX^{-1}}^{1/2}u.$$ 
Take any $x$ and $u,v$ from the tangent space at $x$, i.e. $T_x(\Omega) \cong \R^n$. We have:
$$\inangle{J(x)(u), J(x)(v)} = \inangle{(CX^{-1})^{1/2}u, (CX^{-1})^{1/2}v}=u^\top CX^{-1} v=u^\top H(x)v= \inangle{u,v}_x.$$

\noindent
Thus the inner product $\inangle{\cdot, \cdot}_x$ can be seen as a pull-back of the standard inner product via the map $F$. This also means that $F$ is an isometry: it preserves angles and distances. Let us call the space obtained by the transformation $x\mapsto F(x)$ the $y$-space\footnote{Of course formally $F(\Omega)=\Omega$, hence the $x$-space and $y$-space coincide as sets, the difference however lies in the geometry.}. The point $x$ from the $x$-space is represented by $y=2\sqrt{Cx}$, in the $y$-space. The inverse mapping is $x=\frac{1}{4}C^{-1}y^2$. The strictly feasible set $\Omega^f=\{x\in \R^n:Ax=b,x>0\}$ corresponds to:
$$M=\inbraces{y\in \R^n: \frac{1}{4}AC^{-1}y^2=b, y>0}$$
in the $y$-space.

\subsection{Two Forces}
We intend to show that the direction $P(x)=W(A^\top L^{-1} b- c)$ of the Physarum dynamics~\eqref{eq:dynamics} at the point $x$ decomposes naturally into two directions 
\begin{align}
P_f(x) & \defeq WA^\top \inparen{AWA^\top}^{-1}\inparen{b-Ax}\\
P_o(x) & \defeq W\inparen{A^\top \inparen{AWA^\top}^{-1}Ax-c}
\end{align}
with $P_f(x)$ interpreted as the {\it feasibility direction} and $P_o(x)$ as the {\it optimization direction}. As one can easily see $P(x)=P_f(x)+P_o(x)$. We will interpret both directions by studying the optimization problem arising after transforming the linear program~\eqref{lp} into the $y$-space via the map $F$.

\parag{Feasibility.}
Fix any $\bar{x}\in \Omega$ and consider its image $\bar{y}$ in the $y$-space. Recall that the strictly feasible region of~\eqref{lp} in the $y$-space is $$M=\inbraces{y>0: \frac{1}{4}AC^{-1}y^2=}.$$ Let us assume $\bar{y}\notin M$. In which direction do we need to move, to hit $M$ as soon as possible? This does not seem to be a simple problem, because it is non-convex. However, let us try to guess a $y^\star$ which satisfies:
$$\norm{\bar{y}-y^\star} = \min\inbraces{\norm{\bar{y}-y}: y\in M}$$
and set $h=y^\star-\bar{y}$. Since the above minimization is done with respect to the Euclidean distance, we expect the vector $h$ to be orthogonal to the ``surface'' $M$ at $y^\star$, equivalently to the tangent space $T_{y^\star}M$ i.e. to every vector $u$ which is tangent to the space $M$ at the point $y^\star \in M$. It is easy to see that:
$$T_{y^\star}M = \inbraces{u\in \R^n: \frac{1}{2} AC^{-1}Y^\star u=0}.$$
Since $h$ is just the orthogonal projection of $y^\star - \bar{y}$ onto $T_{y^\star}M$, we can obtain it by the well-known formula\footnote{The orthogonal projection onto the row space of a matrix $B$ is given by $B^\top(BB^\top)^{-1}B$.}:
\begin{align*}
h &= Y^\star C^{-1} A^\top \inparen{A(Y^\star C^{-1})^2 A^\top}^{-1} A C^{-1} Y^\star\inparen{y^\star - \bar{y}}\\
&= Y^\star C^{-1} A^\top \inparen{A(Y^\star C^{-1})^2 A^\top}^{-1} (AC^{-1} (y^\star)^2 - AC^{-1} Y^\star \bar{y})\\
&= Y^\star C^{-1} A^\top \inparen{A(Y^\star C^{-1})^2 A^\top}^{-1} (4b - AC^{-1} Y^\star \bar{y}).
\end{align*}
We do not know what $y^\star$ is, but drawing from the expectation that $\bar{y}$ is close to feasibility, we can approximate $y^\star\approx \bar{y}$, hence obtaining:
$$h\approx \bar{Y} C^{-1} A^\top \inparen{A\inparen{\bar{Y} C^{-1}}^2 A^\top}^{-1} \inparen{4b - AC^{-1} \bar{y}^2}.$$
Let us pull back $h$ to the $x$-space\footnote{More accurately: from the space $T_{\bar{y}}(\Omega)$ to the space $T_{\bar{x}}(\Omega)$.} by substituting $\bar{y}=2\sqrt{C\bar{x}}$ and multiplying $h$ by the inverse-Jacobian $J^{-1}(\bar{x}) = \diag{\frac{\bar{x}_1}{c_1}, \ldots, \frac{\bar{x}_n}{c_n}}^{1/2}=W^{1/2}$. We get:
$$2WA^\top \inparen{AWA^\top}^{-1}\inparen{b-A\bar{x}}=2P_f(\bar{x})$$
which is the feasibility direction up to the a scaling factor $2$.

\parag{Optimality.}
As in the previous discussion fix $\bar{x}\in \Omega$, we will interpret the optimization direction at $\bar{x}$. We know that $\bar{x}$ may not satisfy $A\bar{x}=b$, but we still can write  $A\bar{x}=b^{\bar{x}}$, where $b^{\bar{x}}=b-(b-A\bar{x})$ which we should imagine to be ``close'' to $b$. Our goal is to find the minimum value of $c^\top x$ over the set $\{x: x\geq 0, Ax=b\}$. Let us consider a related linear program:
\begin{align}\label{lp_translated}
\mbox{min. } \quad & c^\top x\\ \nonumber
\mbox{s.t. }\quad & Ax=b^{\bar{x}}\\ \nonumber
& x\geq 0. 
\end{align}
which is just minimization of $c^\top x$ over an affine subspace parallel to $Ax=b$ which contains the point $\bar{x}$. Let us rewrite~\eqref{lp_translated} in the $y$-space.
\begin{align}\label{lp_translated_y}
\mbox{min. } \quad & \frac{1}{4}1^\top y^2 \\ \nonumber
\mbox{s.t. }\quad & \frac{1}{4}AC^{-1}y^2=b^{\bar{x}}\\ \nonumber
& x\geq 0. 
\end{align}
where $1$ denotes the all-one vector. Let $\bar{y}=F(\bar{x})$ and denote 
$$M^{\bar{x}}\defeq \inbraces{y:y>0, \frac{1}{4} AC^{-1}y^2=b^{\bar{x}}}.$$ Note that $\bar{y}\in M^{\bar{x}}$. Let us then compute the steepest descent direction at $\bar{y}$ with respect to the problem~\eqref{lp_translated_y}. We just need to compute the negative gradient of the objective at $\bar{y}$ and project it onto the tangent space $T_{\bar{y}}(M^{\bar{x}})$. We have:
$$\nabla_y\inparen{\frac{1}{4}1^\top y^2} = \frac{1}{2}y.$$
Hence we need to project the vector $-\frac{1}{2}\bar{y}$ onto the space:
$$T_{\bar{y}}(M^{\bar{x}}) = \inbraces{u\in \R^n: \frac{1}{2}AC^{-1}\bar{Y} u=0}.$$
As a result we obtain:
\begin{align*}
u=&\inparen{I-\bar{Y} C^{-1} A^\top \inparen{A(\bar{Y} C^{-1})^2 A^\top}^{-1} A C^{-1} \bar{Y}}\inparen{-\frac{1}{2}\bar{y}}\\
=&\frac{1}{2}\inparen{\bar{Y} C^{-1} A^\top \inparen{A(\bar{Y} C^{-1})^2 A^\top}^{-1} A C^{-1} \bar{y}^2-\bar{y}}.
\end{align*}
To get a corresponding direction in the $x$-space, we compute:
\begin{align*}
J^{-1}(\bar{x})(u) &= \frac{1}{2}W^{1/2}\inparen{\bar{Y} C^{-1} A^\top \inparen{A\inparen{\bar{Y} C^{-1}}^2 A^\top}^{-1} A C^{-1} \bar{y}^2-\bar{y}}\\
&= \frac{1}{2}W^{1/2}\inparen{\frac{1}{2}W^{1/2} A^\top \inparen{AW A^\top}^{-1} A C^{-1} 4C\bar{x}-2\sqrt{C\bar{x}}}\\
&=W\inparen{A^\top\inparen{AWA^\top}^{-1}A\bar{x}-c}=P_o(\bar{x}).
\end{align*}
Which is the optimization direction at $\bar{x}$. Note that in case when $x$ is feasible, the feasibility direction $P_f(x)$ is zero, hence the following theorem.

\begin{theorem}\label{thm:manifold}
Consider the manifold: $\Omega^f = \{x: Ax=b, x>0\}$ endowed with the Riemannian metric $\inangle{u,v}_x = u^\top CX^{-1} v$ at every $x\in \Omega^f$. Then, the Physarum direction $P(x) = W\inparen{A^\top \inparen{AWA^\top}^{-1} b - c}$ is the gradient of the objective function $c^\top x$ with respect to the Riemannian manifold $\inparen{\Omega^f, \inangle{\cdot, \cdot}_x}$.
\end{theorem}
\begin{proof}
One can use the above derivation of the optimization direction to deduce this result. If $x$ is feasible then: $$P_o(x)=W\inparen{A^\top\inparen{AWA^\top}^{-1}Ax-c} =W\inparen{A^\top\inparen{AWA^\top}^{-1}b-c}=P(x).$$
\end{proof}

\subsection{Physarum as an Entropy Barrier Method}\label{sec:entropy}
In this section we show that every Physarum trajectory starting from a strictly feasible point can be seen as a path of optimizers to a certain family of convex programs. Let us fix any starting point $s>0$ satisfying $As=b$ and consider for every $\mu\geq 0$ the following convex program:
\begin{align}\label{eq:reg_convp}
\mbox{min. } \quad & \mu c^\top x + \sum_{i=1}^n x_i c_i \ln (x_i c_i) - \sum_{i=1}^n x_ic_i (1+\ln (s_i c_i)) \\ \nonumber
\mbox{s.t. }\quad & Ax=b\\ \nonumber
& x\geq  0. 
\end{align}
Note that one can rewrite the objective function equivalently as:
$$c^\top x +\frac{1}{\mu}f^s(x).$$
Which is just our original linear objective, regularized by an entropy-like strongly convex function $$f^s(x)\defeq \sum_{i=1}^n x_i c_i \ln (x_i c_i) - \sum_{i=1}^n x_ic_i (1+\ln (s_i c_i)).$$ Observe that $\nabla^2 f^s(x) = CX^{-1}$ is the Riemannian metric we are studying above. It should be intuitively clear that as $\mu\to \infty$  the solution to~\eqref{eq:reg_convp} will tend to the optimal solution of the linear program~\eqref{lp}.

\begin{theorem}\label{thm:entropy}
Take any starting point $s>0$ with $As=b$. Suppose that $x(\mu)$ is the unique optimal solution to~\eqref{eq:reg_convp}. Then $x(\mu)$ for $\mu \in [0,\infty)$ is the solution to the Physarum dynamical system with the initial condition $x(0)=s$.
\end{theorem}

\noindent
The idea of the proof is to write down KKT conditions for the convex program~\eqref{eq:reg_convp} and take advantage of strong duality. By implicit differentiation we conclude that $x(\mu)$ follows the same dynamics as Physarum, furthermore $x(0)=s$, hence they have to coincide. 

\begin{proof}
To improve readability, we will assume that $c=1$, i.e. $c_i=1$ for all $i=1,\ldots, n$. To obtain the general version one can go back and forth with the substitution $x\mapsto C^{-1} z$.

Fix any $s>0$ such that $As=b$. Let us denote $f^s(x)=\sum_{i} x_i \ln x_i  -\sum_{i} x_i (1+\ln s_i )$. First one needs to show that in fact for every $\mu\geq 0$ there exists a minimizer of $\mu 1^\top x + f^s(x)$ in the set $\{x: Ax=b,x>0\}$, in other words that $x(\mu)>0$. This can be seen as follows: pick any point $\bar{x}$ on the boundary, i.e. having $\bar{x}_i=0$ for some $i$. Consider the following scalar function:
$$h(p)=\mu 1^\top (\bar{x}+p(s-\bar{x})) + f^s(\bar{x}+p(s-\bar{x})).$$
for $p\in [0,1]$; $h(p)$ is continuous on $[0,1]$ and differentiable in the interval $(0,1)$. We have:
$$\frac{d}{dp}h(p)=\mu 1^\top (s-\bar{x}) + (s-\bar{x})^\top(\ln(\bar{x}+p(s-\bar{x}))-\ln s).$$
This implies that $h'(p) \to -\infty$ as $p \to 0$, by looking at a coordinate $i$ where $\bar{x}_i=0$. Hence $h(\eps)<h(0)$ for $\eps>0$ small enough, implying (together with convexity) that the optimal value of $\mu 1^\top x + f^s(x)$ over $\{x:Ax=b, x\geq 0\}$ is attained at some point $x(\mu)>0$.

Fix $\mu\geq 0$, introduce dual variables $y\in \R^m$ for the equality constraints and consider the Lagrangian:
$$L(x,y)=\mu 1^\top x + \sum_{i} x_i \ln x_i  - \sum_{i} x_i (1+\ln s_i ) - y^\top(Ax-b).$$
We always keep in mind that $x>0$. Fix $y\in \R^m$ and let us compute the derivative of $L(x,y)$:
$$\nabla_x L(x,y) = \mu 1+(\ln x - \ln s) - A^\top y.$$
the derivative is $0$ at a point $x$ given by:
\begin{align*}
\ln x_i &= a_i^\top y -\mu + \ln s_i\\
x_i & = s_i\cdot  \exp(a_i^\top y - \mu)
\end{align*} 
Then the dual objective $g(y)$ is given by substituting the above $x$ (at which $L(x,y)$ is minimized) in $L(x,y)$. After some cancellations we get:
$$g(y) = y^\top b - \sum_i x_i = y^\top b - \sum_i s_i \cdot \exp(a_i^\top y - \mu )$$
Hence the dual program is:
\begin{align*}
\mbox{min } \quad &  y^\top b - \sum_i s_i \cdot \exp(a_i^\top y - \mu )\\ 
\mbox{s.t. }\quad & y\in \R^m
\end{align*}
Note that strong duality holds, since Slater condition is satisfied for the primal (the witness being $s$). This means that the unique maximizer of g(y) will also yield the optimal solution for the primal. Let 
\begin{align*}
y(\mu) &\defeq  \mbox{argmax}\{y^\top b - \sum_i s_i \cdot \exp(a_i^\top y - \mu ): y\in \R^m\},\\
x_i(\mu) &\defeq s_i\cdot  \exp(a_i^\top y(\mu) - \mu ).
\end{align*}
Then $(x(\mu),y(\mu))$ is the optimal primal-dual pair. 

We will show that $x(\mu)$ is a solution to the dynamical system:
$$\frac{d}{d\mu}x(\mu)=P(x(\mu))=X\inparen{A^\top \inparen{AXA^\top}^{-1} AX1 - 1},$$
which represents Physarum when $c=1$.
 Let us start by examining the initial condition. If $\mu=0$ then one can easily check, that the optimal pair is $(x(0),y(0)) = (s,0)$, as claimed. Let us now compute the derivative of $x(\mu)$ w.r.t. $\mu$, note that the implicit function theorem guarantees that $(x(\mu), y(\mu))$ is a continuously differentiable function of $\mu$. In the following calculations we write shortly $x,y$ for $x(\mu)$ and $y(\mu)$. We have:
$$\dot{x}_i = s_i \cdot \exp(a_i^\top y - \mu) \cdot (a_i ^\top \dot{y} - 1) = x_i (a_i^\top \dot{y} - 1).$$
In other words $\dot{x} = X(A^\top \dot{y} - 1)$. Let us now take the equality $Ax=b$ and differentiate w.r.t. $\mu$:
\begin{align*}
A\dot{x} &= 0\\
AX(A^\top \dot{y} - 1) &=0
\end{align*}
In consequence, $\dot{y}$ satisfies $(AXA^\top) \dot{y}=AX1$. Since $\rank(A)=m$, one can show that $AXA^\top$ is invertible, hence $\dot{y}$ is uniquely determined by $\dot{y} = (AXA^\top)^{-1} AX1$. Consequently:
$$\dot{x} =  X\inparen{A^\top \dot{y} - 1} = X\inparen{A^\top \inparen{AXA^\top}^{-1} AX1 - 1}=P(x).$$ 
\end{proof}

\section{Existence of Solution}\label{section:existence}

In this section we would like to argue that no matter what the initial condition $x(0)>0$ is, the Physarum dynamics has a global solution $x:[0,\infty)\to \R_{>0}^n$. Let us start by explaining what is the significance of the assumption $x(t)>0$. The Physarum dynamics is defined as a dynamical system $\dot{x} = P(x)$, where $P(x) = W(A^\top (AWA^\top)^{-1}b - c)$. For $x>0$ one can show that  $(AWA^\top)^{-1}b$ exists, (since the kernels of $AWA^\top$ and $A^\top$ coincide) hence $P(x)$ is well defined and continuous over $\Omega = \R_{>0}^n$. 
Can we extend it out of $\Omega$? This leads to troubles, because in general $AWA^\top$ might be singular, when $w$ contains some zeros or negative entries. For example if $w=0$ then clearly $AWA^\top =0$, hence there is no hope to extend the vector field to $x=0$. In general there is no simple way to extend (continuously) this vector field from $\Omega=\R_{>0}^n$ to a larger set.

\subsection{Basic Results}
This section introduces some preparatory results which are essential for studying the existence of solution and then asymptotic properties of the Physarum dynamics. Recall that we are always assuming that the feasible region $\{x:Ax=b, x\geq0\}$ is nonempty. 

\begin{remark}\label{remark:vertices}

In the proofs we will repeatedly use the fact that polyhedra of the kind $\{x:Ax=b, x\geq 0\}$ always have at least one vertex. Let $\bar{x}$ be such a vertex. Denote $Z=\{j\in [n]: \bar{x}_j=0\}$ and $N=[n]\setminus Z$. Then $\bar{x}_N$ (i.e. the subvector of $\bar{x}$ consisting of entries $x_j$ for $j\in N$) can be determined as a unique solution to the linear system $A_Nx_N=b$, where $A_N$ is a submatrix of $A$ consisting of columns $a_j$ for $j\in N$. The above are standard facts about vertices of polyhedra and can be found in any book on linear programming, e.g.~\cite{Karloff91}.

\end{remark}

Below we present a key lemma, which then allows to get a solid grasp on the behavior of the vector field $P(x)$ defining Physarum dynamics. 

\begin{lemma}\label{lemma:key}
Let $w\in \R^{n}_{>0}$, $W=\diag{w}$, $L=AWA^\top$. There exists a constant $\alpha>0$ depending only on $A$ such that for every $i \in [n]:$ $$\norm{A^\top L^{-1}a_i}_\infty\leq \frac{\alpha}{w_i}. $$
Quantitatively, one can take $\alpha=\mdet= \max\{|\det(A')|:A' \mbox{ is a square submatrix of A}\}$.
\end{lemma}

\begin{proof}
Fix $i$ and denote by $p$ the solution to the system $Lp=a_i$. We can assume that $p^\top a_j \geq 0$ for every $j\in [m]$ by replacing the row $a_j$ by $-a_j$ if necessary. One can easily see that such a change does not alter the problem, because $L$ remains the same.

Let us first show that $a_i ^\top L^{-1} a_i \leq \frac{1}{w_i}$. Note that
$$w_i a_i a_i^\top\preceq \sum_{j=1}^m w_j a_j a_j ^\top =L$$
where $\preceq$ is the PSD order. This means that $w_i u^\top a_i a_i^\top u \leq u^\top L u$, for every $u\in \R^n$. Let us pick $u=L^{-1} a_i$. We get:
\begin{align*}
w_i u^\top a_i a_i^\top u &\leq u^\top L u \\
w_i a_i^\top L^{-1} a_i a_i^\top L^{-1} a_i &\leq a_i^\top L^{-1}  L L^{-1} a_i \\ 
w_i (a_i^\top L^{-1} a_i)^2 &\leq a_i^\top L^{-1} a_i\\
a_i^\top L^{-1} a_i &\leq \frac{1}{w_i}.
\end{align*}

\noindent
It remains to argue that for some $\alpha$ and for every $k\in [n]:$ $$a_k^\top p\leq \alpha a_i^\top p.$$ 
Fix $k\in [n]$, assume $k\neq i$. If $a_k^\top p=0$ then we are done, assume $a_k^\top p>0$. From $Lp=a_i$ we get:
$$ \sum_{j=1}^m w_j a_j (a_j ^\top p) = a_i.$$
Hence the set $S_k \defeq \{s\in \R^m _{\geq 0}: As = a_i \wedge s_k>0\}$ is nonempty ( $WA^\top p$ belongs to it). Take $s\in S_k$ with $s_k$ maximum possible (it can be seen that $s_k$ is bounded over all $s\in S_k$). Then $\sum_{j=1}^m s_j a_j =a_i$, hence $\sum_{j=1}^m s_j a_j ^\top p = a_i ^\top p$. Since $s_j a_j^\top p\geq 0$ for all $j$, we can deduce that $s_k a_k^\top p \leq a_i^\top p$ and hence $a_k^\top p \leq \frac{a_i^\top p}{s_k}=\alpha_k a_i^\top p$. It is enough to choose $\alpha=\max_{k} \alpha_k$. 

For the quantitative bound one needs to note that $\alpha$ is chosen according to the following values: $\eps_k = \max\{s_k : As=a_i, s\geq 0\}$. In fact $\alpha = \max_{k} \frac{1}{\eps_k}$. Because linear programs attain optimal values in vertices, one can argue that $s^\star$ -- the optimal solution to $\max\{s_k : As=a_i, s\geq 0\}$ (for some fixed $k$) can be chosen to be a vertex of the polyhedron $\{s:As=a_i, s\geq 0\}$. By the Cramer's rule, every positive entry of $s^\star$ is lower-bounded by $\mdet^{-1}$.
\end{proof}
\noindent Let us now see what happens when $w=C^{-1}x$ actually comes from a feasible vector.
\begin{corollary}\label{cor:pot_diff}
Suppose that $x>0$ and $Ax=b$, we have:
$$\norm{A^\top p}_\infty\leq \mdet \cdot C_s. $$
\end{corollary}
\begin{proof}
Let $w=C^{-1} x$, using Lemma~\ref{lemma:key} we get:
\begin{align*}
\norm{A^\top p}_\infty &=\norm{A^\top L^{-1}b}_\infty = \norm{\sum_{i=1}^n A^\top L^{-1}(a_i x_i)}_\infty \\
& \leq \sum_{i=1}^n x_i \norm{A^\top L^{-1}a_i }_\infty\\
& \leq \sum_{i=1}^n x_i \frac{\mdet}{w_i}= \mdet \cdot C_s .
\end{align*}
\end{proof}

\noindent It turns out that if we could prove a result such as above for all $x\in \Omega$, not only for feasible $x$, it would imply the existence of solution. Unfortunately, this ceases to be true when do not make this restriction. Intuitively, our strategy is to prove that $\norm{A^\top p}$ is bounded over every trajectory $\{x(t):0\leq t\leq T\}$, this then allows us to reason that no trajectory will leave $\Omega$. The next lemma shows a uniform bound but with $\norm{A\top p}$ replaced by $\norm{WA^\top p}$.
\begin{lemma}\label{lemma:qbound}
Suppose $y$ if feasible: $Ay=b$ and $y\geq 0$. Then for every $x\in \R^n_{>0}$ it holds that $\norm{q}_{\infty} \leq \alpha \norm{y}_1$. As in Lemma~\ref{lemma:key} we can take $\alpha=\mdet$.
\end{lemma}
\begin{proof}
Recall that $q=WA^\top L^{-1}b$, hence $q_i = w_i a_i^\top L^{-1} b$. We express $b$ as $b=Ay = \sum_{j=1}^n  y_j a_j$.
\begin{align*}
|q_i|&= \abs{w_i a_i^\top L^{-1} \inparen{ \sum_{j=1}^n  y_j a_j}}
= \abs{\sum_{j=1}^n  y_j w_i a_i^\top L^{-1} a_j}\\
&\leq \sum_{j=1}^n w_i y_j \abs{a_i ^\top L^{-1} a_j}
=\sum_{j=1}^n w_i y_j \abs{a_j ^\top L^{-1} a_i}\\
&\stackrel{\mbox{Lemma \ref{lemma:key}}}{\leq} \sum_{j=1}^n w_i y_j \frac{\alpha}{w_i} = \alpha \norm{y}_1
\end{align*}
Hence $\norm{q}_\infty \leq \alpha \norm{y}_1$.
\end{proof}
\noindent The following corollary gives a concrete bound we can deduce from the above lemma.

\begin{corollary}\label{cor:qbound}
For every $x\in \R_{>0}^n$ it holds that $\norm{q}_\infty \leq \mdet^2 \cdot n \cdot \norm{b}_1$.
\end{corollary}
\begin{proof}
We use Lemma~\ref{lemma:qbound}. Let us pick $y$ as any vertex of the feasible region $\{x:Ax=b, x\geq 0\}$. We know that the non-zero portion $y_N$ of y can be determined as a solution to the linear system $A_N y_N=b$. Hence, by Cramer's rule every non-zero entry of $y$ ca be obtained as follows:
$$y_i =\sum_{j=1}^m b_j \frac{\alpha_j}{\alpha_0}.$$
for some $\alpha_0, \alpha_1, \ldots, \alpha_m$ being determinants of square submatrices of $A$. Since $A$ has integer entries, $|\alpha_0|\geq 1$ and $|\alpha_1|, \ldots, |\alpha_m|\leq \mdet$, thus $|y_i|\leq \norm{b}_1 \mdet$ and $\norm{y}_1 \leq n \norm{b}_1 \mdet$.
\end{proof}

\subsection{The Proof of Existence}
We are now ready to prove the following theorem (same as Theorem \ref{thm:existence0}). As previously, we will use $\Omega$ to denote $\R_{>0}^n$.
\begin{theorem}\label{thm:existence}
For every initial condition $x(0) \in \Omega $ the dynamical system~\eqref{eq:dynamics} has a unique solution $x:[0,\infty) \to \Omega $.
\end{theorem}
\noindent
From now on we assume that the initial condition $x(0)=s\in \Omega $ is fixed.
In the light of Theorem~\ref{theorem:bap_existence}, to prove~\ref{thm:existence} it suffices to show the following two lemmas. 
\begin{lemma}\label{lemma:bap1}
Suppose $x:[0,T)\to \Omega$ is a solution to~\eqref{eq:dynamics} for some $T\in \R_{>0}$ then $\lim_{t \to T^{-}} x(t)$ exists.
\end{lemma}
\begin{lemma}\label{lemma:bap2}
Suppose $x:[0,T)\to \Omega$ is a solution to~\eqref{eq:dynamics} for some $T\in \R_{>0}$. If $x(T)\defeq \lim_{t \to T^{-}} x(t)$ exists then $x(T)\in \Omega$, i.e. $x(T)$ is strictly positive.
\end{lemma}

\noindent Let us start by proving that every trajectory stays in a bounded region.
\begin{lemma}\label{cor:xbound}
Suppose $x:[0,T)\to \Omega$ is a solution to \eqref{eq:dynamics}, then for every $i\in [n]$ and $t\in [0,T)$:
$$x_i(t)\leq \max\inparen{x_i(0),\beta},$$
where $\beta = \mdet^2 \cdot n \cdot \norm{b}_1$.
\end{lemma}

\begin{proof}
Fix $i\in [n]$, by Corollary~\ref{cor:qbound} we have:
$$\dot{x}_i = q_i - x_i\leq \beta - x_i.$$
Applying Gronwall lemma to $y(t) = x_i(t) - \beta$ yields:
$$ x_i(t) \leq (1-e^{-t}) \beta + e^{-t} x_i(0)\leq \max\inparen{x_i(0),\beta}.$$
\end{proof}

\noindent We are now ready to prove Lemma~\ref{lemma:bap1}:
\begin{proofof}{of Lemma~\ref{lemma:bap1}}
Suppose $x:[0,T)\to \Omega$ is a solution to \eqref{eq:dynamics}. Fix $i\in [n]$, we show that $x_i(t)$ has a limit with $t\to T^{-}$. We know that $x_i$ is differentiable in the interval $[0,T)$ and $|\dot{x}_i|$ is bounded, since 
$$|\dot{x}_i(t)|=|q_i(t)-x_i(t)|\leq |q_i(t)|+|x_i(t)| \leq \beta + \max(x_i(0), \beta)$$
by Lemma~\ref{lemma:qbound} and Corollary~\ref{cor:xbound}. Hence $x_i(t)$ is a Lipschitz function in the interval $[0,T)$. This implies that $\lim_{t \to T^{-}} x_i(t)$ exists.
\end{proofof}

Now we are left with the task of proving that no coordinate of $x$ will approach $0$ as $t\to T^{-}$. This task simplifies significantly (as implied by the next lemma) if we assume that there is a strictly feasible solution to~\eqref{lp}, i.e. $y\in \R^n$ such that $Ay=b$ and $y>0$. Nevertheless, we do not want to make any such assumptions to make our proof valid in full generality. We proceed by showing that positivity holds for every $i\in \supp(y)$, for any feasible $y$, here $\supp(y)$ is the support of the vector $y$, i.e. the set $\{j\in [n]:y_j>0\}$.

\begin{lemma}\label{lemma:supp_pos}
Suppose $x:[0,T) \to \Omega$ is a solution to \eqref{eq:dynamics} and $y$ is any feasible solution to \eqref{lp}. If $x(T) = \lim_{t\to T^{-}} x(t)$ then $x_i(T)>0$ for every $i\in \supp(y)$.
\end{lemma}
\begin{proof}
Fix any feasible solution $y$ to \eqref{lp}. To justify the claim, we use the following ``barrier function'':
$$\barr(t) \defeq \sum_{j=1}^n y_j c_j \ln x_j(t).$$
$\barr(t)$ is clearly well defined on $[0,T)$. Suppose for the sake of contradiction that $\inf_{t \in [0,T)} x_{i}(t) = 0$ for some $i\in \supp(y)$. We know by Corollary~\ref{cor:xbound} that $x_j(t)$ are uniformly upper-bounded over $t\in [0,T)$. Hence we know that $y_jc_j \ln x_j(t) \leq M$ for all $j\in [n]$, $t\in [0,T)$ and some constant $M\in \R$. It follows that $\inf_{t \in [0,T)} \barr(t) = -\infty$. This implies that the derivative $\dot{\barr}(t)$ is unbounded from below, for $t\in [0,T)$.

However, let us compute this derivative:
\begin{equation}
\dot{\barr}(t) = \sum_{j=1}^n y_jc_j \frac{\dot{x_j}(t)}{x_j(t)}=\sum_{j=1}^ny_j c_j \frac{q_j(t)}{x_j(t)}-\sum_{j=1}^ny_j c_j=\sum_{j=1}^n y_j c_j \frac{q_j(t)}{x_j(t)}-c^\top y.
\end{equation}
We analyze the term $\sum_{j=1}^ny_j c_j \frac{q_j(t)}{x_j(t)}$. 
$$\sum_{j=1}^ny_j c_j \frac{q_j(t)}{x_j(t)} = \sum_{j=1}^ny_j c_j \frac{w_j a_j^\top p(t) }{x_j(t)}=\sum_{j=1}^ny_j a_j^\top p(t) =(Ay)^\top p(t) = b^\top p(t) .$$
Note that $b^\top p = b^\top L^{-1}b \geq 0$, hence the above sum is nonnegative. In consequence:
$$\dot{\barr}(t) \geq -c^\top y=\mbox{const.}$$
which is a contradiction, since this quantity was claimed to be unbounded from below.
\end{proof}

We have now all necessary tools to prove Lemma~\ref{lemma:bap2}.
\begin{proofof}{of Lemma \ref{lemma:bap2}}
Let $x:[0,T) \to \Omega$ be a solution to \eqref{eq:dynamics}. Fix $y$ to be any feasible solution to \eqref{lp}. Lemma~\ref{lemma:supp_pos} implies that there exists $\eps>0$ such that $\eps\cdot  y \leq x(t)$ for every $t\in [0,T)$. Indeed: for $i$ such that $y_i>0$ we have $\inf_{t\in [0,T)} x_i(t)>0$, so $\eps_i \cdot y_i \leq x_i(t)$ for some $\eps_i>0$ and all $t\in [0,T)$, then take $\eps=\min\{\eps_i: i\in \supp(y)\}$.

Fix any $t\in [0,T)$ and consider $w\in \R^n$ defined as $w_i = \frac{x_i(t)}{c_i}$. Let us bound the norm of $A^\top L^{-1} b$:

\begin{align*}
 \norm{A^\top L^{-1} b}_\infty 
&=  \norm{A^\top L^{-1} \inparen{\sum_{j=1}^n y_j a_j}}_\infty\\
&\leq \sum_{j=1}^n y_j \norm{A^\top L^{-1} a_j}_\infty
 \stackrel{\mbox{Lemma \ref{lemma:key}}}{\leq} \sum_{j=1}^n y_j \frac{\alpha}{w_j} \\
&\leq \sum_{j=1}^n \frac{x_j(t)}{\eps} \frac{\alpha}{w_j} = \frac{\alpha}{\eps} \sum_{j=1}^n  c_j.
\end{align*}
Let us denote $M = \frac{\alpha}{\eps} \sum_{j=1}^n  c_j$. Pick now any $i\in [n]$, we have:
$$\dot{x}_i =q_i - x_i = \frac{x_i}{c_i} p_i - x_i \geq x_i \inparen{-\frac{M}{c_i}-1}.$$
Hence, from the Gronwall Lemma we obtain:
$$x_i(t) \geq x_i(0) e^{t\inparen{-\frac{M}{c_i}-1}}.$$
In particular $x_i(t)$ is lower bounded by a positive constant over the whole interval $[0,T)$. The lemma follows.
\end{proofof}

\noindent We can now conclude Theorem~\ref{thm:existence} easily.

\begin{proofof}{of Theorem~\ref{thm:existence}}
We use Theorem~\ref{theorem:bap_existence}. One needs to observe that the function $P(x)=q-x$ is in fact of class $\mathcal{C}^1$. Even more is true: $P(x)$ is a rational function in the region $\Omega$, this can be seen using the Cramer's rule. The remaining assumptions are satisfied by Lemmas~\ref{lemma:bap1}and \ref{lemma:bap2}.
\end{proofof}

\section{Asymptotic Behavior of Physarum Trajectories}\label{section:convergence}
In the previous section we have established the existence of solutions to the Physarum dynamics. We know that for every starting point $s>0$ there is a solution $x:[0,\infty) \to \R_{>0}^n$ to the Physarum dynamics with $x(0)=s$. We would like to study the behavior of $x(t)$ when $t\to \infty$. Let us fix the starting point $x(0)$ and pick $\maxx>0$ such that:
$$\maxx^{-1} \leq x_i(0)\leq \maxx \qquad \mbox{ for every }i\in [n].$$
Furthermore, pick $x^\star$ to be any optimal basic solution (i.e. $x^\star$ is a vertex of the feasible region) to~\eqref{lp}. Our first result will be that $c^\top x(t)$ approaches $c^\top x^\star$ -- the optimal value of the linear program~\eqref{lp}. Before we proceed with the proof let us establish some useful properties of the feasible region.

\subsection{The Feasible Region}
Let us consider $P=\{x:Ax=b,x\geq0\}$ -- the feasible region of~\eqref{lp}. Note that $P$ does not contain a line, so it can be expressed as the Minkowski sum $P=H+K$, where $H$ is the convex hull of vertices of $P$ and $K$ is a polyhedral cone $$K=\inbraces{r\in \R^n: Ar=0, r\geq 0}.$$ Let us denote by $V$ the set of vertices of $P$, so $H=\conv(V)$.
Further define $R$ to be the set of vertices of the polytope $$\inbraces{r:Ar=0,\; \;  \sum_{i=1}^n r_i = 1, r\geq 0}.$$ Then $K$ is the conic hull of $R$: 
$$K=\inbraces{\sum_{r\in R} \mu_r r: \mu_r \geq 0}.$$ Let us first establish the following bounds:
\begin{lemma}\label{lemma:vertices}
Let $V$ and $R$, be as defined above. 
\begin{enumerate}
\item Let $v\in V$ and $i\in [n]$ such that $v_i\neq 0$, then $\mdet^{-1} \leq \abs{v_i} \leq \mdet \cdot \norm{b}_1$.
\item Let $r\in R$ and $i\in [n]$ such that $r_i\neq 0$, then $\mdet^{-1} \leq \abs{r_i} \leq \mdet$.
\end{enumerate}
\end{lemma}
\begin{proof}
This follows immediately from the Cramer's rule and the characterization of vertices of polyhedra (see Remark~\ref{remark:vertices}).
\end{proof}

Let $x\in P$, since $P=H+K$, $x$ can be expressed as:
$$x=\sum_{v\in V} \lambda_v v + \sum_{r\in R} \mu_r r$$
with $\lambda_v, \mu_r\geq 0$ for $v\in V, r\in R$ and $\sum_{v\in V} \lambda_v=1$. Note that the sets $V$ and $R$ might be very large (of size exponential in $n$). However, by Caratheodory's theorem we may pick the vectors of coefficients $\lambda, \mu$ so that $|\supp(\lambda)|\leq n+1$ and $|\supp(\mu)| \leq n+1$.

The last preliminary fact we would like to state is the following proposition on sensitivity analysis of linear programs. It is very useful to estimate the distance from a polyhedron to a point which satisfies its constraints with some additive error.
\begin{proposition}[Sensitivity analysis]\label{prop:sensitivity}
Suppose that $B\in \Z^{M \times N}$, $g\in \R^N$ and $b', b''\in \R^M$ are such that both linear programs: $\min\{g^\top x: Bx\leq b'\}$ and $\min\{g^\top x: Bx\leq b''\}$ have finite optimal values, then:
$$\abs{\min\{g^\top x: Bx\leq b'\} - \min\{g^\top x: Bx\leq b''\}} \leq N  \mdet_B \norm{g}_1 \cdot \norm{b'-b''}_\infty.$$
Where $\mdet_B = \max\{|\det(B')|: B' \textnormal{ a square submatrix of }B\}$.
\end{proposition}
\noindent
For a proof of the above proposition we refer to~\cite{Schrijver86}.

\subsection{Convergence to the Optimal Value}
In this section we establish the following theorem.
\begin{theorem}\label{thm:exp}
Suppose that $x:[0,\infty) \to \Omega$ is any solution to the Physarum dynamics. Then, for some $R,\nu>0$ depending only on $A,b,c,x(0)$, we have:
$$\abs{c^\top x(t) - c^\top x^\star} \leq R \cdot e^{-\nu t}  .$$
Quantitatively, one can take $\nu=\mdet^{-3}$ and $R=\exp(8\mdet ^2 \cdot C_s\cdot \norm{b}_1)\cdot (n+\maxx)^2$.
\end{theorem}
Let us begin by splitting the set $V$ into two subsets $V_O=\{v\in V: c^\top v = c^\top x^\star\}$ and $V_N= V\setminus V_O$. Then $V_O$ is basically the set of optimal vertices. The following bound will be very useful in the proof of Theorem~\ref{thm:exp}.
\begin{lemma}\label{lemma:cost_dist}
Let $v\in V_N$, then $c^\top v - c^\top x^\star \geq \mdet^{-2}$.
\end{lemma}
\begin{proof}
Since $v$ is a vertex of $P$, by Cramer's rule it can be expressed as $v=(\frac{z_1}{d}, \frac{z_2}{d}, \ldots, \frac{z_n}{d})^\top$. Where $d,z_1, z_2, \ldots,z_n \in \Z$ and $1\leq d\leq \mdet$. Hence $c^\top v$ is a number of the form $\frac{z}{d}$ for some integer $z$. Similarly $c^\top x^\star = \frac{z'}{d'}$ with $z',d'\in \Z$ and $1\leq d' \leq \mdet$. We get:
$$c^\top v - c^\top x^\star = \frac{d'z - z'd}{dd'}\geq \frac{1}{\mdet^2}.$$
\end{proof}

The next lemma says that $x(t)$ becomes exponentially close to the feasible region $P$ as $t\to \infty$.
\begin{lemma}\label{lemma:feas}
For every $t\geq 0$ there exists $y(t)\in P$ such that: $$\norm{x(t)-y(t)}_\infty<e^{-t}\cdot 3n^2 \cdot \mdet \cdot \maxx.$$
\end{lemma}
\begin{proof}
Start with the Physarum dynamics $\dot{x} = q-x$ and multiply both sides by $e^t$. We get $\frac{d}{dt}(x(t) e^t) = e^t q(t)$. Take integrals of both sides to obtain:
$$x(t) = x(0)e^{-t}+(1-e^{-t})\int_0^t q(s) \frac{e^{s-t}}{1-e^{-t}}ds.$$
Observe that $z(t)=(1-e^{-t})\int_0^t q(s) \frac{e^{s-t}}{1-e^{-t}}ds$ is a convex combination of $q(s)$ (which satisfy $Aq(s)=b$ for every s), hence we conclude that $z(t)=x(t) - x(0)e^{-t}$ satisfies $Az(t)=b$. Note also that $z(t)\geq x(0)e^{-t}$, so $z(t)$ approaches $P$ with $t\to \infty$. We would like to estimate the distance from $z(t)$ to the closest point in $P$.
As a tool for that we would like to use Proposition~\ref{prop:sensitivity}. We will use variables $y\in \R^n$ and a new variable $d\in \R$. Let us write down two linear programs:

\begin{equation*}
\begin{aligned}[c]
\mbox{min. } \quad & d\\ 
\mbox{s.t.  }\quad & Ay=b\\ 
& \norm{y-z(t)}_\infty \leq d \\
& y\geq 0
\end{aligned}
\qquad \qquad \qquad
\begin{aligned}[c]
\mbox{min. } \quad & d\\ 
\mbox{s.t.  }\quad & Ay=b\\ 
& \norm{y-z(t)}_\infty \leq d \\
& y\geq -x(0)e^{-t}.
\end{aligned}
\end{equation*}

\noindent The minimum value of the first linear program is equal to $\dist^\infty(z(t),P)$: the minimum $\ell^{\infty}$--distance from $z(t)$ to some point $y\in P$. The second one has optimal value $0$, which is demonstrated by taking $y=z(t)$. To apply Proposition~\ref{prop:sensitivity} we need too transform the constraints into inequality form, but this is done in a standard manner. We obtain:
$$\dist^\infty(z(t),P) \leq (n+1)\cdot(2n)\cdot \mdet \cdot \norm{x(0) e^{-t}}_\infty.$$
In the above, $(2n)\cdot \mdet$ is the bound on $\mdet_B$ where $B$ is the matrix obtained after transforming the constraints into the inequality form. By taking $y(t)$ which attains this minimum distance, we obtain the desired point in $P$:
\begin{align*}
\norm{x(t) - y(t)}_\infty &\leq \norm{x(t)-z(t)}_\infty + \norm{z(t) - y(t)}_\infty\\
 &\leq \norm{x(0)e^{-t}}_\infty + (n+1)\cdot(2n)\cdot \mdet \cdot \norm{x(0)e^{-t}}_\infty\\
 &\leq 3n^2 \cdot \mdet \cdot \norm{x(0)e^{-t}}_\infty.
\end{align*}
\end{proof}
\noindent
By the discussion at the beginning of the section we know that for every $t$, $y(t)$ (from the above lemma) can be decomposed as:
\begin{equation}\label{eq:dec}
y(t) = \sum_{v\in V} \lambda_{v}(t) v + \sum_{r \in R}\mu_{r}(t) r
\end{equation}
for some $\lambda_v(t), \mu_r(t) \geq 0$ such that $\sum_{v\in V} \lambda_v(t) = 1$. Of course this decomposition is not unique. We choose it arbitrarily, but as noted previously it can be done in such a way that most of the coefficients are zeros, i.e. $|\supp(\lambda(t))|\leq n+1$ and $|\supp(\mu(t))|\leq n+1$. Our strategy is to show that $\mu_r(t) \to 0$ for all $r\in R$ and $\lambda_v(t) \to 0$ for all $v\in V_{N}$. This will imply that $y(t)$ (hence also $x(t)$) tends to the optimal region. Towards this let us prove a preparatory result.

\begin{lemma}\label{lemma:conv0}
The following bounds hold with some $Q, \nu>0$ depending only on $A,b,c,x(0)$:
\begin{enumerate}
\item For every $r\in R$: $\min\{x_i(t): i\in \supp(r)\}\leq Q \cdot \exp(-\nu t)$.
\item For every $v\in V_N$: $\min\{x_i(t): i\in \supp(v)\}\leq Q \cdot \exp(-\nu t)$.
\end{enumerate}
Quantitatively, one can take $\nu=\mdet^{-3}$ and  $Q = \exp\inparen{4\mdet^2 C_s \cdot \norm{b}_1}\cdot \inparen{n+\maxx}$.
\end{lemma}
\begin{proof}
Let us start with the second part. Pick any $v\in V_N$ and $x^\star \in V_O$, consider the following potential function:
$$f(t) \defeq \sum_{i=1}^n c_i v_i \ln x_i(t) - \sum_{i=1}^n c_i x_i^\star \ln x_i(t)$$
and take its derivative:
\begin{align*}
\frac{d}{dt}f(t) &=  \sum_{i=1}^n c_i v_i \frac{\dot{x}_i(t)}{x_i(t)}- \sum_{i=1}^n c_i x_i^\star \frac{\dot{x}_i(t)}{x_i(t)}\\
& =  \sum_{i=1}^n c_i v_i \inparen{\frac{a_i^\top p(t)}{c_i}-1}- \sum_{i=1}^n c_i x_i^\star \inparen{\frac{a_i^\top p(t)}{c_i}-1}\\
& = \sum_{i=1}^n \inparen{v_i - x_i^\star}a_i^\top p(t) + \inparen{c^\top x^\star - c^\top v}.
\end{align*}
Observe that, since $Av=Ax^\star=b$, we have:
$$\sum_{i=1}^n \inparen{v_i - x_i^\star}a_i^\top p(t) = (v-x^\star)A^\top p(t) = (b^\top - b^\top )p(t) = 0.$$
\noindent
Hence, the derivate takes a particularly simple form:
$$\frac{d}{dt}f(t) = c^\top x^\star - c^\top v.$$
By Lemma~\ref{lemma:cost_dist} we know that $\frac{d}{dt}f(t) \leq - \eps$ for $\eps = \mdet^{-2}$. This implies the following bound on $f(t)$:
$$f(t) \leq f(0) - \eps \cdot t.$$
Hence $$\sum_{i=1}^n c_i v_i \ln x_i(t) \leq f(0) + \sum_{i=1}^n c_i x_i^\star \ln x_i(t) - \eps \cdot t.$$
\noindent
Let us use Corollary~\ref{cor:xbound} and Lemma~\ref{lemma:vertices} to bound:
$$\sum_{i=1}^n c_i x_i^\star \ln x_i(t) \leq \sum_{i=1}^n c_i \mdet \cdot \norm{b}_1 \ln (\max(\mdet^2 n \norm{b}_1, \maxx))\leq C_s \mdet \cdot \norm{b}_1\cdot \ln(\mdet^2n\norm{b}_1+\maxx).$$
Let us call $M$ the whole expression on the right-hand side. Similarly we can bound $f(0) \leq 2M$. We get:
$$\sum_{i=1}^n c_i v_i \ln x_i(t) \leq 3M - \eps t.$$
Suppose wlog that $x_1(t) = \min\{x_i(t): i\in \supp(v)\}$, then:
$$\ln x_1(t)\cdot \sum_{i\in \supp(v)} c_i v_i\leq  \sum_{i=1}^n c_i v_i \ln x_i(t) \leq 3M - \eps t.$$
It remains to bound $\sum_{i\in \supp(v)} c_i v_i \geq \mdet ^{-1}$ to get:
$$\min\{x_i(t): i\in \supp(v)\} \leq \exp(3M \mdet)\cdot \exp(-\mdet^{-3}t).$$
The proof of the first part is analogous, but uses the potential function $\sum_{i=1}^n c_i r_i \ln x_i(t)$.
\end{proof}

\noindent We are now ready to present a complete proof of the Theorem:
\begin{proofof}{of Theorem~\ref{thm:exp}}
Let us first pick $r\in R$, we will show that $\mu_r(t) \to 0$ exponentially fast. We have for every $i\in \supp(r)$:
$$\mu_r(t) r_i \leq y_i(t) \leq x_i(t)+M e^{-t}.$$
with $M$ as in Lemma~\ref{lemma:feas}. Hence, by Lemma~\ref{lemma:conv0}, there is some $i\in \supp(r)$ such that:
$$\mu_r(t) r_i\leq Q e^{-\nu t}+M e^{-t}.$$
By the bound $r_i\geq \mdet^{-1}$ from Lemma~\ref{lemma:vertices} we obtain:
$$\mu_r(t) \leq \mdet Q e^{-\nu t}+\mdet M e^{-t} \leq 2 \mdet Q e^{-\nu t}.$$
By a similar argument, we get for every $v\in V_N$:
$$\lambda_v(t) \leq 2 \mdet Q e^{-\nu t}.$$
Having this, let us consider the difference: $c^\top y(t) - c^\top x^\star$. We know that:
$$c^\top \inparen{\sum_{v\in V_O} \lambda_v(t) v}=  \sum_{v\in V_O} \lambda_v(t)c^\top v = (c^\top x^\star)\sum_{v\in V_O} \lambda_v(t).$$
Hence we get a bound:
$$c^\top y(t) - c^\top x^\star \leq c^\top \inparen{\sum_{v\in V_N} \lambda_v(t)v + \sum_{r\in R} \mu_r(t) r}.$$
Recall that $\lambda(t)$ and $\mu(t)$ were chosen in such a way that all but at most $n+1$ of their entries are zero. Hence:
$$\norm{\sum_{v\in V_N} \lambda_v(t)v}_\infty \leq \sum_{v\in V_N} \lambda_v(t) \norm{v}_\infty \leq (n+1)\cdot 2 \mdet Q e^{-\nu t} \cdot \mdet \cdot \norm{b}_1.$$
A similar bound holds for $\sum_{r\in R} \mu_r(t) r$. Altogether, we obtain the following bound:
$$c^\top y(t) - c^\top x^\star \leq 4(n+1)C_s \mdet^2 \norm{b}_1 Q e^{-\nu t}.$$

To conclude let us relate $c^\top y(t)$ to $c^\top x(t)$. We have, by Lemma~\ref{lemma:feas}:
$$\abs{c^\top y(t)-c^\top x(t)} \leq C_s \cdot Me^{-t}.$$
This finally yields the claimed bound:
$$\abs{c^\top x(t) - c^\top x^\star} \leq 4(n+1)C_s \mdet^2 \norm{b}_1 Q e^{-\nu t} + C_s \cdot Me^{-t}\leq Q^2 e^{- \nu t}.$$
\end{proofof}

\subsection{Convergence to a Limit}
By the previous subsection we know that for every solution $x:[0,\infty)\to \Omega$, $\abs{c^\top x(t) - c^\top x^\star} \to 0$ exponentially fast with $t\to \infty$. One can use this fact to prove that whenever~\eqref{lp} has a unique solution then $x(t)$ actually has a limit $x^\infty$ and $x^\infty$ is this unique optimal solution to~\eqref{lp}. However, if we do not assume uniqueness the fact that $x(t)$ has a limit is no more obvious. We are aiming to prove it in the current subsection. Let us now formally state the theorem.
\begin{theorem}\label{thm:limit}
Suppose $x:[0,\infty)\to \Omega$ is a solution to the Physarum dynamics~\eqref{eq:dynamics}. Then there exists a limit $x^\infty = \lim_{t\to \infty} x(t)$. Furthermore $x^\infty$ is an optimal solution to the linear program~\eqref{lp}. 
\end{theorem}
In this subsection the norm symbol $\norm{\cdot}$ should be understood as the infinity norm $\norm{\cdot}_\infty$. Unlike in the previous proofs we will not keep track of constants and hide them under the big O notation. By a constant we mean any quantity which solely depends on $A,b,c,n,m$ and the fixed starting point $x(0)$ under consideration.

Let us start by giving a simple lemma which provides a sufficient condition for existence of a limit.
\begin{lemma}\label{lemma:limit}
Suppose $x:[0,\infty)\to \Omega$ is a differentiable function. If the integral:
$$\int_{0}^\infty \norm{\dot{x}(t)}dt$$
is finite, then there exists a limit $x^\infty = \lim_{t\to \infty} x(t)$.
\end{lemma}
\begin{proof}
Observe that $x(t) = x(0)+\int_{0}^t \dot{x}(s) ds$. By going with $t\to \infty$, we obtain:
$$\lim_{t\to \infty} x(t) = x(0)+\int_{0}^\infty \dot{x}(s)ds.$$
Since the integral on the right-hand side exists (because it is absolutely convergent by the assumption), we conclude existence of the limit.
\end{proof}
In our case $\dot{x}=q-x$. Our proof strategy is to show that $\norm{q(t)-x(t)}=O(e^{-\eps t})$ for some constant $\eps>0$, then the assumption of Lemma~\ref{lemma:limit} is satisfied.

Let us now review what do we know from the proof of convergence to the optimal value. The following fact can be extracted from the proof of Theorem~\ref{thm:exp}:
\begin{lemma}\label{lemma:close_opt}
For every $t\geq 0$ there exists $x^\star(t)$ -- an optimal solution to~\eqref{lp} such that $\norm{x(t) - x^\star(t)}=O(e^{-\eps t})$, for some $\eps>0$.
\end{lemma}
\noindent
Note that the above lemma does not imply that $x(t)$ in fact has a limit. It could potentially happen that $x(t)$ oscillates over the optimal region without approaching a single point. To prove it we need to understand the behavior of Physarum with respect to the structure of the optimal set. 

From the general theory of linear programming (see e.g.~\cite{Wright97}) we know that $[n]$ can be decomposed into $[n]=J\cup N$ with $J\cap N=\emptyset$ such that every optimal solution to~\eqref{lp} is supported on $J$ and every feasible solution supported on $J$ is optimal. We will denote by $x_J$ the part of a vector $x$ corresponding to indices $j\in J$, similarly $x_N$. Let us now establish two important lemmas, which we can prove using techniques developed so far.
\begin{lemma}\label{lemma:xN}
If $N$ is the above define set of indices then $\norm{x_N(t)}=O(e^{-\eps t})$ for some $\eps>0$.
\end{lemma}
\begin{proof}
It follows immediately from~\ref{lemma:close_opt}. For every $t$, $x^\star_N(t)=0$, because $x^\star(t)$ is optimal, hence the claim.
\end{proof}
\begin{lemma}\label{lemma:xJ}
There exists a constant $\eps>0$ such that $x_j(t)>\eps$ for every $j\in J$ and $t\in [0,\infty)$.
\end{lemma}
\begin{proof}
Recall that the set of optimal vertices is denoted by $V_O$. Note that for every $j\in J$, there is a vertex $v\in V_O$ such that $v_j = \Omega(1)$. Hence, it is enough to show that for every $v\in V_O$, the function:
$$f_v(t) = \sum_{i=1}^n c_i v_i \ln x_i(t)$$
is bounded from below by a universal constant (this is true because of the fact that $\norm{x(t)}$ is uniformly bounded, by Corollary~\ref{cor:xbound}). 
Note that, by a calculation analogous to that performed in the proof of~\ref{lemma:conv0}, we obtain:
$$\frac{d}{dt}(f_v(t) - f_u(t))=0$$
for every $v,u\in V_O$, which implies that all $f_v$ differ only by a constant. Hence either all of them are lower-bounded by a universal constant or none of them. 

By Lemma~\ref{lemma:close_opt} we know that for every $t$, $x(t)$ is exponentially close to some $x^\star(t)\in \conv(V_O)$. This implies in particular that for big enough $t$ and for some universal constant $\delta>0$ there always exists a $v(t)\in V_O$ such that $\delta \cdot v(t) \leq x(t)$. Hence:
$$f_{v(t)}(t) \geq \sum_{i=1}^n c_i v_i(t) \ln (\delta \cdot v_i(t)) = \Omega(1).$$
Note that the set $V_O$ is finite, hence we get the bound $\max_{v\in V_O} f_v(t)=\Omega(1)$, which finishes the proof.
\end{proof}

The last preparatory lemma, which can be deduced from previous considerations is the following:
\begin{lemma}\label{lemma:pbound}
Suppose $p(t)$ is the vector $L^{-1}b$, computed with respect to $x(t)$. Then $\norm{A^\top p(t)}$ is uniformly bounded over $t\in [0,\infty)$.
\end{lemma}
\begin{proof}
By a reasoning as in the proof of Lemma~\ref{lemma:xJ} we get that for some $v\in V$ and a universal constant $\delta>0$ we have $\delta\cdot v \leq x(t)$ for every $t\in [0,\infty)$. Hence we get:
$$\norm{A^\top p(t)} = \norm{A^\top L^{-1} b} = \norm{A^\top L^{-1}\inparen{\sum_{i=1}^n a_i v_i}}\leq \sum_{i=1}^n v_i \norm{A^\top L^{-1} a_i}.$$
By Lemma~\ref{lemma:key}, we have that $\norm{A^\top L^{-1} a_i} \leq \frac{\alpha c_i}{x_i}$, for some universal constant $\alpha$. Hence, we get:
$$\sum_{i=1}^n v_i \norm{A^\top L^{-1} a_i} \leq \sum_{i=1}^n v_i \frac{\alpha c_i}{x_i} \leq  \sum_{i=1}^n (\delta^{-1} x_i) \frac{\alpha c_i}{x_i} = \alpha \delta^{-1} \sum_{i=1}^n c_i =  O(1).$$
The lemma follows.
\end{proof}

We now present the main technical lemma required for proving existence of a limit. Intuitively it proves that the subvector of $A^\top p$ corresponding to $J$ behaves continuously assuming $x$ is close to the optimal set.

\begin{lemma}\label{lemma:continuity}
Suppose $\eps>0$ is small enough and $M, d>0$. Pick any $x>0$ such that $\norm{x_N}<\eps$ and $\norm{x_J - y_J}<\eps$ for some optimal solution $y$, such that $y_j>d$ for all $j\in J$. Assume $\norm{x}\leq M$ and $\norm{A^\top p}\leq M$. There exists a constant $N$ depending only on $A,b,c,d,n$ such that $\norm{q-x}<N\cdot \eps$.
\end{lemma}
\begin{proof}

We will denote the columns of $A$ corresponding to $J$ by $A_J$, similarly for $A_N$. We have $(AWA^\top )p = b$, which can be rewritten as:
$$A_JW_J A_J^\top p + A_N W_N A_N^\top p = b.$$ 
where $W_J=\diag{w_J}$ and $W_N=\diag{w_N}$.
Since $\norm{A_N^\top p}=O(1)$ and $\norm{A_NW_N}=\norm{A_N C^{-1} X_N}=O(\eps)$ we have $\norm{A_N W_N A_N^\top p}=O(\eps)$. Moreover we can write $x_J = y_J + \tau_J$, where $\norm{\tau_J}<\eps$. This gives a decomposition of $W_J$ into $\Wb=\diag{C_J^{-1} y_J}$ and $\Wt=\diag{C_J^{-1} \tau_J}$ such that $\norm{A_J \Wt A_J^\top p}=O(\eps)$. In consequence:
$$\norm{A_J \Wb A_J^\top p - b}\leq \norm{A_J W_J A_J^\top p - b}+\norm{\Wt A_J^\top p}=O(\eps).$$
Pick $\bar{p}$: any optimal solution to the dual of the linear program~\eqref{lp}. From complementary slackness we have $Y  (A^\top \bar{p} - c)=0$. In particular $A_J^\top \bar{p}=c_J$. We obtain:
$$A_J \Wb A_J^\top \bar{p} = A_J \Wb  c_J = A_J \diag{C_J^{-1} y_J }c_J = A_Jy_J = b .$$
Hence:
$$\norm{A_J \Wb A_J^\top (p- \bar{p})}=\norm{A_J \Wb A_J^\top p - A_J \Wb A_J^\top \bar{p}}=\norm{A_J \Wb A_J^\top p - b}=O(\eps).$$
We want to show that $\norm{A_J^\top (p- \bar{p})}=O( \eps)$. 

\noindent Denote $K=A_J \Wb A_J^\top$, it is a symmetric PSD matrix. Moreover the kernels of $A_J^\top$ and $K$ coincide, since $y_J>0$. 

\noindent Pick a vector $v\in \R^m$ so that $A_J^\top (p- \bar{p}) = A_J^\top v$ and $v$ is orthogonal to the kernel of $A_J^\top$. (In other words $v$ is the orthogonal projection of $p-\bar{p}$ onto the orthogonal complement of the kernel of $A_J^\top$.) 

\noindent Using the fact that $\frac{y_i}{c_i}>\frac{d}{c_i}\geq d'$ for $i\in J$ and some absolute positive constant $d'$, we get:
$$d' A_J  A_J^\top=\sum_{j\in J}d' a_j a_j^\top \preceq \sum_{j\in J}\frac{y_j}{c_j} a_j a_j^\top =A_J \Wb A_J^\top = K$$ 
where $\preceq$ denotes the PSD ordering. This implies in particular that $d'\lambda^+ \leq \lambda_K^+$, where $\lambda^+, \lambda_K^+$ are the smallest positive eigenvalues of $A_J  A_J^\top$ and $K$ respectively. 
Since $v$ is orthogonal to the kernel of $K$ and $K$ is PSD we have:
$$\norm{Kv}_2 \geq \lambda_K^+ \norm{v}_2 \geq d'\lambda^+ \norm{v}_2$$
Observe that $\lambda^+$ depends solely on $A,b,c$, hence:
$$\norm{v} \leq \norm{v}_2 \leq \frac{1}{d'\lambda^+} \norm{Kv}_2 = \frac{1}{d'\lambda^+} \norm{K(p-\bar{p})}_2=O(\eps).$$
In consequence $\norm{A^\top _J (p-\bar{p}) }= \norm{A^\top _J v}=O(\eps)$ and hence $\norm{\Wb A^\top _J (p-\bar{p}) }=O(\eps)$, because $\norm{y}=O(1)$ (the set of optimal solution is bounded). Finally:
$$\norm{W_J A_J^\top p - \Wb A_J^\top \bar{p}}\leq \norm{W_J A_J^\top p - \Wb A_J^\top p}+\norm{ \Wb A_J^\top p -  \Wb A_J^\top \bar{p}}=O(\eps)+O(\eps) = O(\eps).$$
But  $W_J A_J^\top p = q_J$ and $\Wb A_J^\top \bar{p}=y_J$, which means:
$$\norm{q_J - x_J} \leq \norm{q_J -y_J}  + \norm{y_J - x_J} =O(\eps)+O(\eps) = O(\eps).$$
To conclude, note that:
$$\norm{q_N-x_N}\leq \norm{x_N}+\norm{q_N}= \norm{x_N}+\norm{W_N A_N^\top p }= O(\eps)$$
and hence trivially: 
$$\norm{q-x} = \max(\norm{q_N - x_N}, \norm{q_J - x_J}) = O(\eps).$$
\end{proof}
\noindent
Let us now conclude Theorem~\ref{thm:limit} from Lemma~\ref{lemma:continuity}.
\begin{proofof}{of Theorem~\ref{thm:limit}}
By Lemma~\ref{lemma:limit}, it remains to show that $\norm{q(t)-x(t)} =O(e^{-\delta t})$ for some $\delta>0$. By Corollary~\ref{cor:xbound} and Lemma~\ref{lemma:pbound} we obtain a uniform bound $M$ on both $\norm{x(t)}$ and $\norm{A^\top p(t)}$ over $t\in [0,\infty)$. From Lemma~\ref{lemma:xJ} we obtain some $d>0 $ such that $x_j(t)>2d$ for $j\in J$ and $t\in [0,\infty)$. 

Let us now pick a large $t$. We take $x=x(t)$ and $y=x^\star(t)$ from Lemma~\ref{lemma:close_opt}. We know that $x_j>2d$ for $j\in J$, $y$ is exponentially close to $x$, hence $y_j>d$ for $j\in J$. Moreover, we can pick $\eps = O(e^{-\delta t})$ such that $\norm{x_N}<\eps$ (by Lemma~\ref{lemma:xN}) and $\norm{y-x}<\eps$. Hence the assumptions of Lemma~\ref{lemma:continuity} are satisfied and we may conclude that $\norm{q(t)-x(t)}<N\eps = O(e^{-\delta t})$. The theorem follows.
\end{proofof}

\section{Discretization}\label{section:discretization}
In this section we study the efficiency of the natural discretization of the Physarum dynamics. It is defined with respect to the step length $0<h<1$ and a starting point $x(0)>0$. In this section, we will always assume that the starting point was chosen from the feasible region, i.e. $Ax(0)=b$. The discrete dynamics is as follows:
\begin{equation}\label{eq:discrete}
x(k+1) = (1-h)x(k)+hq(k)
\end{equation}
with the usual meaning of $q(k)$. For the above dynamics to be well defined we need to make sure that $x(k)$ stays positive for all $k\geq 0$. Note that:
$$x_i(k+1) =x_i(k)\inparen{1-h\inparen{\frac{a_i^\top p(k)}{c_i} -1}}$$
hence we need to make sure that $\inparen{\frac{a_i^\top p(k)}{c_i} -1}<\frac{1}{h}$. Corollary~\ref{cor:pot_diff} tells us that:
$$\px \defeq \max\inbraces{\abs{\frac{a_i^\top p}{c_i}-1}:Ax=b, x\geq 0}\leq C_s \mdet+1.$$
This means that it is enough to set $h< \px^{-1}$ (and we do so in the following theorem).
\begin{theorem}\label{theorem:discrete_conv}
Let $0<\eps<\nfrac{1}{2}$, $h\leq \px^{-2} \cdot \nfrac{\eps}{6}$. Suppose we initialize the Physarum algorithm~\eqref{eq:discrete} with $x(0)$, s.t. $Ax(0)=b$ and $\maxx^{-1} \leq x_i(0)\leq \maxx$ for every $i\in [n]$ and some $\maxx \geq 1$. Assume additionally that $c^\top x(0) \leq M\cdot \opt$. Then after $k\defeq O\inparen{\frac{\ln M}{\eps^2 h^2} + \frac{\ln \maxx}{\eps h}}$ steps $x(k)$ is a feasible solution with: $\opt \leq c^\top x(k) \leq (1+\eps)\opt$.
\end{theorem}

\noindent To prove Theorem~\ref{theorem:discrete_conv} we analyze the following potential function. It is analogous to that used in the paper~\cite{SV15}.
\begin{equation}\label{eq:phi_potential}
	\phi(k) \defeq4 \ln \pot(k) - \frac{\eps h }{\opt}\barr(k).
\end{equation}
The intuitive meaning of $\phi$ is as follows: we know that $\pot(k)$ is non-increasing with $k$, so $\ln \pot(k)$ is non-increasing as well, however we may not guarantee a big drop of $\pot(k)$ in every step, this is why we have the second term; $\barr(k)$ gets bigger at such steps. The crucial lemma we would like to show asserts that $\phi$ drops significantly at every step:

\begin{lemma}[Potential drop]\label{lemma:potential_drop}
For every k with $\pot(k)>(1+\eps)\opt$ we have $\Delta \phi(k) \leq -\frac{h^2\eps^2}{6}$.
\end{lemma} 
\noindent Before proving the lemma we first conclude Theorem~\ref{theorem:discrete_conv} from it.

\begin{proofof}{of Theorem~\ref{theorem:discrete_conv}  assuming Lemma~\ref{lemma:potential_drop}}
We will first estimate the value of $\phi(0)$. Observe that
$$\ln(\pot(0)) = \ln (c^\top x(0)) \leq \ln (M\cdot \opt).$$
The value $\barr(0)$ can be bounded as follows:
$$\sum_{i=1}^n x_i^\star c_i \ln x_i(0) \geq x_i^\star c_i \ln \maxx^{-1} = \opt \cdot \ln \maxx^{-1}.$$
Hence we obtain:
$$\phi(0)\leq 4\ln (M\cdot \opt) + \eps h \ln \maxx.$$
\noindent
Let us now fix $k$ and provide a lower bound on the value of $\phi(k)$. We know that $\pot(k)\geq \opt$ for every $k$, for $\barr$ we have:
$$\barr(k) = \sum_{i=1}^n x_i^\star c_i \ln x_i(k) \leq \ln (\maxx) \sum_{i=1}^n x_i^\star c_i = \opt \cdot \ln (\maxx), $$
and thus:
$$\phi(k) \geq 4\ln (\opt) - \eps h \ln (\maxx).$$
Putting together the above estimates on $\phi(0), \phi(k)$ and Lemma~\ref{lemma:potential_drop}, we obtain that:
$$k \frac{h^2 \eps^2}{6}\leq  \phi(0)-\phi(k)\leq 4 \ln M + \eps h( \ln (\maxx)+\ln (\maxx)).$$
Simplifying:
$$k = O\inparen{\frac{\ln M}{\eps^2 h^2} + \frac{\ln \maxx}{\eps h}}.$$
\end{proofof}

We split the Lemma~\ref{lemma:potential_drop} into two cases and deal with them separately.  They are expressed in the following facts. We use below $\ene(k)$ to denote $q(k)^\top W^{-1} q(k)$.

\begin{fact}\label{fact:big_gap}
If $\frac{\ene(k)}{\pot(k)}<(1-\eps/3)$ then $\Delta \phi(k) \leq -\frac{h^2\eps^2}{6}$.
\end{fact}

\begin{fact}\label{fact:small_gap}
If $\ene(k)>(1+\eps/3)\opt$ then $\Delta \phi(k) \leq -\frac{h^2\eps^2}{6}$.
\end{fact}
\noindent
Before we proceed with the proof of facts, let us state three simple inequalities, which we are going to use:
\begin{align}
\label{eq:log_upper}
\ln(1+\alpha) &\leq \alpha  &\mbox{ for every }& \alpha \in \R \\
\label{eq:log_lower}
\ln(1+\alpha) & \geq \alpha - \alpha^2 &\mbox{ for every }& \alpha \in \insquare{-\frac{1}{2},\frac{1}{2}}\\
\label{eq:log_negative}
\ln(1-\alpha) & \geq -2\alpha  &\mbox{ for every }& \alpha \in \insquare{0,\frac{1}{2}}
\end{align}
All can be proved by simple calculus. Another useful fact is the following identity:
\begin{fact}\label{fact:energy}
Let $x\in \R_{>0}^n$ and $p=(AWA^\top)^{-1}b$ be given. Suppose $y\in \R^n$ satisfies $Ay=b$, then
$$\sum_{i=1}^n y_i a_i^{\top} p = q^\top W^{-1} q$$
where $q=WA^\top p$. 
\end{fact}

\begin{proofof}{of Fact~\ref{fact:big_gap}}
We will show that $\ln \pot$ drops and $\barr$ may increase only a little. Let us first look at $\ln \pot(k+1)-\ln \pot(k)$. We have:
$$\pot(k+1) = \sum_{i=1}^nc^\top x(k+1) = (1-h)c^\top x(k) + h\sum_{i=1}^n x(k)a_i^\top p(k) = (1-h)\pot(k)+h \ene(k).$$
The last inequality follows from~\ref{fact:energy}. Further:
\begin{align*}
\frac{\pot(k+1)}{\pot(k)}&= \inparen{1+h\inparen{\frac{\ene(k)}{\pot(k)}-1} }\\
&<  \inparen{1+h\inparen{\inparen{1-\frac{\eps}{3}}-1} }\\
&= 1-h\frac{\eps}{3} .
\end{align*}
We obtain:
$$\ln \pot(k+1)-\ln \pot(k)=\ln \frac{\pot(k+1)}{\pot(k)}\stackrel{ \eqref{eq:log_upper}}{\leq} -\frac{h\eps}{3}.$$
Consider now $\barr(k+1)-\barr(k)$:
\begin{align*}
\barr(k+1)-\barr(k)&=\sum_{i=1}^n x_i^\star c_i \ln \frac{x_i(k+1)}{x_i(k)}\\
&=\sum_{i=1}^n x_i^\star c_i \ln \inparen{1+h\inparen{\frac{a_i^\top p(k)}{c_i}-1}}\\
&\geq  \sum_{i=1}^n x_i^\star c_i \ln \inparen{1-h\px}\\
&\stackrel{ \eqref{eq:log_negative}}{\geq} -2h(n^2C+1)\sum_{i=1}^n x_e^\star c_e\\
&= -2h\px\cdot \opt\\
&\geq- \opt .
\end{align*}
The last inequality follows from the definition of $h$.
Putting these two pieces together yields:
$$\phi(k+1)-\phi(k)\leq -\frac{4h\eps}{3} + h \eps =-\frac{h\eps}{3}\leq -\frac{h^2\eps^2}{6}.$$
\end{proofof}

\begin{proofof}{of Fact~\ref{fact:small_gap}}
We know that $\pot(k)$ is non-increasing with $k$, it remains to show that $\barr$ will increase by a considerable amount. As in the proof of Fact~\ref{fact:big_gap} we obtain:
$$\barr(k+1)-\barr(k)=\sum_{i=1}^n x_i^\star c_i\ln \inparen{1+h\inparen{\frac{a_i^\top p(k)}{c_i}-1}}.$$
By the choice of $h$:
$$\abs{h\inparen{\frac{q_e(k)-x_e(k)}{x_e(k)}}}=\abs{h}\abs{\frac{a_i^\top p(k)}{c_e}-1}\leq h\cdot \px \leq \frac{1}{2}.$$
Hence we can apply the inequality~\ref{eq:log_lower}, by which we get:
\begin{align*}
\barr(k+1)-\barr(k)&\stackrel{\eqref{eq:log_lower}}{\geq}\sum_{i=1}^n x_i^\star c_ih\inparen{\frac{a_i^\top p(k)}{c_i}-1}-\sum_{i=1}^n x_i^\star c_ih^2\inparen{\frac{a_i^\top p(k)}{c_i}-1}^2\\
&\geq h\sum_{i=1}^n x_i^\star a_i^\top p(k) - h\cdot \opt -\sum_{i=1}^n x_i^\star c_i\inparen{h\px}^2\\
&\stackrel{\mathrm{Fact } \ref{fact:energy}}{=} h\ene(k)- h\cdot \opt-h^2\opt \px^2\\
&\geq  h\ene(k) - h\cdot \opt -\frac{1}{6}h\eps \opt \\
&\geq  h\cdot \opt\inparen{1+\frac{\eps}{3}} - h\cdot \opt -\frac{1}{6}h\eps \opt \\
&\geq-\frac{1}{6} h\cdot \eps\cdot  \opt .
\end{align*}
This implies the following drop of $\phi$:
$$\phi(k+1) - \phi(k) \leq - \frac{h\eps}{\opt}\inparen{\barr(k+1)-\barr(k)}\leq -\frac{h^2\eps^2}{6}.$$
\end{proofof}

\begin{proofof}{of Lemma~\ref{lemma:potential_drop}}
If $\eps$ is small enough and $\pot(k)>(1+\eps)\opt$ then obviously either $\frac{\ene(k)}{\pot(k)}<(1-\eps/3)$ or $\ene(k)>(1+\eps/3)\opt$ and we use Fact~\ref{fact:big_gap} or Fact~\ref{fact:small_gap} respectively.
This concludes our proof.
\end{proofof}

\bibliographystyle{alpha}
\bibliography{references}

\appendix

\section{Dynamical Systems and Global Existence}\label{section:dynamical}
We work in the $n$-dimensional Euclidean space $\R^n$, an open subset $\Omega \subseteq \R^n$ is chosen for the domain of the dynamical system. The equations are of the form:
$$\dot{x}_i(t)=F_i(x_1(t),x_2(t),\ldots,x_n(t))\qquad \mbox{for $i=1,2,\ldots,n$},$$
where $x_1, x_2, \ldots, x_n$ are functions of single variable $t$ regarded as unknowns and $F_1, F_2, \ldots, F_n:\Omega \to \R$ are given. The above system is often denoted shortly as:
$$\dot{x} = F(x)$$
where $F:\Omega \to \R^n$ is just $F=(F_1, F_2, \ldots, F_n)$. Typically we impose some initial condition of the form $x(0)=s$, with $s\in \Omega$. A solution to such a system is a function $x:I\to \Omega$, where $I$ is some non-degenerate interval containing $0$, $x(0)=s$ and $\dot{x}(t)=F(x(t))$ for every $t$ in the interior of $I$. We list some typical questions which arise in the study of a dynamical system. Suppose we are given a dynamical system $S$.
\begin{enumerate}
\item Is there any solution to $S$? What is the maximal interval of existence?
\item Is the solution unique?
\item What are the local properties of $S$ (around point $x(0)$)?
\item What are the global properties of the solution? Does it converge? Does it stay in a bounded region?
\end{enumerate}
We are interested mainly in the last question, in both qualitative and quantitative aspects (whether it does converge and what is the convergence rate). However, to attempt question (4) we must first answer (1) and (2); otherwise we may end up studying objects which do not even exist. Moreover, in many cases existence (and uniqueness\footnote{Note that even if a differential equation has a solution it may be not unique. Consider for example the following: $\dot{x}=3x^{2/3}$ with $x(0)=0$. It has two solutions $x(t)=t^3$ and $x(t)=0$.}) can be very hard to establish.

In our case it is not difficult to obtain a solution on some small interval $(-\eps,\eps)$. However, for applications we need existence on $[0,+\infty)$, (we call it global existence) which is harder to establish. Intuitively it may happen that at some finite time step $T$ the solution $x$ will ``escape'' from the domain $\Omega$. More formally, even if the solution $x$ exists in the interval $[0,T)$, the limit $\lim_{t\to T^{-}} x(t)$ may be outside of $\Omega$ or the limit may not even exist.   Escaping from $\Omega$ is not only a formal problem of defining $\Omega$, it is actually a serious issue, because in our case, the function $F$ is not well defined outside of $\Omega$.

In this section we provide a sufficient condition for a dynamical system to have global solutions. In the main body we prove that in fact the Physarum dynamical system satisfies this condition.

We first state an important general theorem which gives sufficient conditions for existence on some interval around $0$. A proof can be found in any textbook on dynamical systems, e.g.~\cite{perko2001differential}.
\begin{theorem}\label{theorem:existence_perko}
Let $\Omega \subseteq \R^n$ be an open set, $s\in \Omega$ and $F:\Omega \to \R^n$ be a function of class $\mathcal{C}^1$. Consider the following dynamical system: $\dot{x} = F(x)$ with initial condition $x(0)=s$. There exists a unique solution $x:(-\eps, \eps)\to \Omega$ to the system, for some $\eps>0$.
\end{theorem}

\noindent We conclude this subsection by proving a theorem that a certain simple condition implies existence on the whole half-line.
\begin{theorem}\label{theorem:bap_existence}
Let $\Omega \subseteq \R^n$ be an open set, $s\in \Omega$ and $F:\Omega \to \R^n$ be a function of class $\mathcal{C}^1$. Consider the following dynamical system: $ \dot{x}= F(x)$ with initial condition $x(0)=s$. Suppose it satisfies the following conditions for every solution $x:[0,T)\to \Omega$ with $T\in (0,\infty)$:
\begin{itemize}
\item the limit $\lim_{t\to T^{-}} x(t)$ exists,
\item if $x(T)=\lim_{t\to T^{-}} x(t)$ then $x(T)\in \Omega .$
\end{itemize}
Then the there exists a global solution $x:[0,\infty)\to \Omega$.
\end{theorem}
\begin{proof}
Let $T_0$ be the maximal $T$ such that a solution $x:[0,T)\to \Omega$ exists (we allow $T_0=\infty$). One should argue why does such a $T_0$ exists. This is a consequence of uniqueness: whenever there are two solutions $x:[0,T_1)\to \Omega$ and $y:[0,T_2)\to \Omega$ with $T_1\leq T_2$, then they agree on $[0,T_1)$.

From Theorem~\ref{theorem:existence_perko} we know that $T_0>0$. We want to show that $T_0=\infty$. For the sake of contradiction assume that $T_0$ is a finite number. Denote by $x(T_0)$ the limit $\lim_{t\to T_0^{-}}x(t)$. By assumption $x(T_0)\in \Omega$, thus we can use Theorem~\ref{theorem:existence_perko} to extend the solution from $[0,T_0)$ to $[0,T_0+\eps)$ for some $\eps>0$. 
This gives us a contradiction with the definition of $T_0$.
\end{proof}

\noindent An example of a system which does not satisfy the the assumption of Theorem~\ref{theorem:bap_existence}.
\begin{example}
Consider a one-dimensional system $\dot{x}=\frac{1}{1-x}$, with $x(0)=0$. The natural choice for the domain is $\Omega = \R \setminus \{1\}$ (only on this set the function $x\mapsto \frac{1}{1-x}$ makes sense). Let $T=\frac{1}{2}$, then there is a solution $x:[0,T) \to \Omega$ given by $x(t) = 1-\sqrt{1-2t}$. However:
$$\lim_{t\to T^{-}} x(t) = 1.$$
But $1\notin \Omega$, hence this system does not satisfy the condition from Theorem~\ref{theorem:bap_existence}. Note also that actually there is no solution to this system on the whole half-line $[0,\infty)$.
\end{example}

\end{document}

%% file: macros.tex
\newtheorem{theorem}{Theorem}[section]

\newtheorem{lemma}[theorem]{Lemma}
\newtheorem{remark}[theorem]{Remark}
\newtheorem{proposition}[theorem]{Proposition}
\newtheorem{corollary}[theorem]{Corollary}

\newtheorem{fact}[theorem]{Fact}

\newtheorem{example}[theorem]{Example}

\def\FullBox{\hbox{\vrule width 6pt height 6pt depth 0pt}}

\def\qed{\ifmmode\qquad\FullBox\else{\unskip\nobreak\hfil
\penalty50\hskip1em\null\nobreak\hfil\FullBox
\parfillskip=0pt\finalhyphendemerits=0\endgraf}\fi}

\def\qedsketch{\ifmmode\Box\else{\unskip\nobreak\hfil
\penalty50\hskip1em\null\nobreak\hfil$\Box$
\parfillskip=0pt\finalhyphendemerits=0\endgraf}\fi}

\newenvironment{proofof}[1]{\begin{trivlist} \item {\bf Proof
#1:~~}}
  {\qed\end{trivlist}}

\newcommand\Z{\mathbb Z}

\newcommand\R{\mathbb R}

\newcommand{\marginlabel}[1]%
{\mbox{}\marginpar{\it{\raggedleft\hspace{0pt}#1}}}
\definecolor{Mygray}{gray}{0.8}

 \ifcsname ifcommentflag\endcsname\else
  \expandafter\let\csname ifcommentflag\expandafter\endcsname
                  \csname iffalse\endcsname
\fi

\ifnum\showauthornotes=1
\newcommand{\todo}[1]{\colorbox{Mygray}{\color{red}#1}}
\else
\newcommand{\todo}[1]{}
\fi

\ifnum\showauthornotes=1
\newcommand{\Authornote}[2]{{\sf\small\color{red}{[#1: #2]}}}
\newcommand{\Authoredit}[2]{{\sf\small\color{red}{[#1]}\color{blue}{#2}}}
\newcommand{\Authorfnote}[2]{\footnote{\color{red}{#1: #2}}}
\newcommand{\Authorfixme}[1]{\Authornote{#1}{\textbf{??}}}
\newcommand{\Authormarginmark}[1]{\marginpar{\textcolor{red}{\fbox{%\Large
#1:!}}}}
\else
\newcommand{\Authornote}[2]{}
\newcommand{\Authoredit}[2]{}
\newcommand{\Authorcomment}[2]{}
\newcommand{\Authorfnote}[2]{}
\newcommand{\Authorfixme}[1]{}
\newcommand{\Authormarginmark}[1]{}
\fi

\newcommand{\inparen}[1]{\left(#1\right)}             %\inparen{x+y}  is (x+y)
\newcommand{\inbraces}[1]{\left\{#1\right\}}           %\inbrace{x+y}  is {x+y}
\newcommand{\insquare}[1]{\left[#1\right]}             %\insquare{x+y}  is [x+y]
\newcommand{\inangle}[1]{\left\langle#1\right\rangle} %\inangle{A}    is <A>

\def\abs#1{\left| #1 \right|}
\newcommand{\norm}[1]{\ensuremath{\left\lVert #1 \right\rVert}}

\newcommand{\diag}[1]{{\sf Diag}\left({#1}\right)}

\newlength{\pgmtab}  %  \pgmtab is the width of each tab in the
 {
	\begin{enumerate}}{\end{enumerate}}

\newcounter{lecnum}

\newlength{\tpush}
\ifnum\showdraftbox=1

\else

\fi